\documentclass[12pt, leqno]{article}

\usepackage{amsmath}
\usepackage{amssymb}
\usepackage{geometry}
\usepackage{amsthm}
\usepackage{amscd}
\usepackage{graphicx}
\usepackage{tikz-cd}
\usepackage{mathtools}
\usepackage[mathscr]{euscript}
\usepackage{setspace}

\usepackage{stmaryrd}

\usepackage{graphicx}
\usepackage{tikz}

\usepackage{fancyhdr}
\usepackage[indentafter]{titlesec}
\titleformat{name=\section}{}{\thetitle.}{0.8em}{\centering\scshape}
\titleformat{name=\subsection}[runin]{}{\thetitle.}{0.5em}{\bfseries}[.]
\titleformat{name=\subsubsection}[runin]{}{\thetitle.}{0.5em}{\itshape}[.]
\titleformat{name=\paragraph,numberless}[runin]{}{}{0em}{}[.]
\titlespacing{\paragraph}{0em}{0em}{0.5em}
\titleformat{name=\subparagraph,numberless}[runin]{}{}{0em}{}[.]
\titlespacing{\subparagraph}{0em}{0em}{0.5em}

\usepackage[obeyspaces]{url}
\usepackage{hyperref}
\hypersetup{
    colorlinks=true,
    linkcolor=blue,
    filecolor=magenta,      
    urlcolor=cyan,
    filecolor=red,
    citecolor=cyan,
}

\graphicspath{ {Desktop/} }
\newtheorem{thm}{Theorem}[section]

\newtheorem{prop}[thm]{Proposition}
\newtheorem{lem}[thm]{Lemma}
\newtheorem{cor}[thm]{Corollary}
\theoremstyle{definition}
\newtheorem{defn}[thm]{Definition}
\newtheorem{ex}[thm]{Example}
\newtheorem{rmk}[thm]{Remark}
\newtheorem{const}[thm]{Construction}

\DeclareMathOperator{\Tr}{Tr}

\DeclareMathOperator{\act}{\ \rotatebox[origin=c]{+90}{$\circlearrowleft$}\ }

\makeatletter
\newcommand*{\doublerightarrow}[2]{\mathrel{
  \settowidth{\@tempdima}{$\scriptstyle#1$}
  \settowidth{\@tempdimb}{$\scriptstyle#2$}
  \ifdim\@tempdimb>\@tempdima \@tempdima=\@tempdimb\fi
  \mathop{\vcenter{
    \offinterlineskip\ialign{\hbox to\dimexpr\@tempdima+1em{##}\cr
    \rightarrowfill\cr\noalign{\kern.5ex}
    \rightarrowfill\cr}}}\limits^{\!#1}_{\!#2}}}
\newcommand*{\triplerightarrow}[1]{\mathrel{
  \settowidth{\@tempdima}{$\scriptstyle#1$}
  \mathop{\vcenter{
    \offinterlineskip\ialign{\hbox to\dimexpr\@tempdima+1em{##}\cr
    \rightarrowfill\cr\noalign{\kern.5ex}
    \rightarrowfill\cr\noalign{\kern.5ex}
    \rightarrowfill\cr}}}\limits^{\!#1}}}
\makeatother

\begin{document}

\begin{center}
\large{\textbf{General Covariance from the Viewpoint of Stacks}} \\

Filip Dul\footnote{Department of Mathematics, Rutgers -- New Brunswick, Piscataway NJ 08854, United States of America. \\ \indent \hspace{0.16cm} ORCID: 0000-0001-8623-0293. \\
\indent \hspace{0.16cm} Author email address: fdul.math@gmail.com \\} 

\end{center}

\begin{abstract}
General covariance is a crucial notion in the study of field theories on curved spacetimes. A field theory defined with respect to a semi-Riemannian metric is generally covariant if two metrics on $X$ which are related by a diffeomorphism produce equivalent physics. From a purely mathematical perspective, this suggests that we try to understand the quotient stack of metrics modulo diffeomorphism: we will use the language of groupoids to do this concretely. Then we will inspect the tangent complex of this stack at a fixed metric, which when shifted up by one defines a differential graded Lie algebra. By considering the action of this Lie algebra on the observables for a Batalin-Vilkovisky scalar field theory, we recover a novel expression of the stress-energy tensor for that example, while describing how this works for a general class of theories. We will describe how this construction nicely encapsulates but also broadens the usual presentation in the physics literature and discuss applications of the formalism. 
\end{abstract}

\indent Batalin-Vilkovisky Formalism, Gravitation, Conserved Quantities.

\tableofcontents

\section{Introduction}
\subsection{Overview} Over a hundred years ago, when Albert Einstein and a group of others were laying the foundations of general relativity, general covariance became an essential ingredient in formulating physics in curved spacetimes. Roughly, a field theory coupled to a background metric on a spacetime $X$ is said to be generally covariant if the diffeomorphism group of $X$ is a symmetry of the theory. Physicists usually interpret diffeomorphisms as coordinate changes, so they may say that a theory exhibits general covariance if it is coordinate-invariant: i.e. a theory may superficially change to a distinct one if the coordinates are changed, but if it is generally covariant, then those ``two" theories are equivalent in a way which we will make rigorous. Although general covariance can be understood in the context of \textit{all} field theories, it is often considered in the context of field theories coupled to semi-Riemannian metrics: this is the \textit{particular case} we will focus on in this paper, but we will moreover argue that this particular case is of central importance.

The primary aim of this paper is to package general covariance in the Batalin-Vilkovisky formalism for classical field theories, especially as it is presented in \cite{cosgwill2}. Equivariant vector bundles are an appropriate geometric tool for understanding families of field theories parameterized by a space of Riemannian or Lorentzian metrics. We therefore review the global theory of such bundles, explain its relevance for general covariance in Section \ref{GCsection}, and then cast much of it in the language of stacks in Section \ref{groupoids}. Stacks provide a generalized notion of space which allows us to deal with quotient spaces that may be singular or otherwise forget interesting information about the original space. We will describe how a bundle of generally covariant classical Batalin-Vilkovisky (BV) field theories over the space of metrics on $X$, denoted $\mathscr{M}$, descends to a bundle over the quotient stack $[\mathscr{M}/\mathscr{D}]$ of the metrics modulo the diffeomorphism group of $X$, denoted $\mathscr{D}$.  Our definition for general covariance is equivalent to the following:

\begin{defn}
A bundle $\pi: (\mathscr{F}, \{S,-\}) \to \mathscr{M}$ of Batalin-Vilkovisky field theories on a compact manifold $X$	 is \textbf{generally covariant} if it descends to a bundle of stacks 
$$
\pi: ([\mathscr{F/D}], \{S,-\}) \to [\mathscr{M/D}].
$$
\end{defn}

Before and after introducing this definition, we discuss in detail how both scalar field theory and Yang-Mills theory are generally covariant in our sense as BV field theories. Indeed, one of our central results is Theorem \ref{YMequivariance}, in which we show that Yang-Mills theory is generally covariant, because it serves as a nexus for views on covariance from other sources in the literature, as we discuss in the surrounding commentary. We make remarks on functorial aspects of our work (mostly in Subsection \ref{functors}) which are important in their own right and also serve as points of comparison to the prevailing literature--both to the factorization algebra framework (as seen in \cite{cosgwill1} and \cite{cosgwill2}) the author was trained in, but also to the AQFT framework (as seen in \cite{rejzner}) the author desires to better understand.

Much of our concrete computations are in the regime of perturbative field theories, so we consider formal neighborhoods in the quotient stack $[\mathscr{M}/\mathscr{D}]$, which in turn lead us to understanding the field theories as examples in derived deformation theory: in brief, we associate to a generally covariant theory a formal moduli problem, as defined in \cite{lurie}, by pulling back the above bundle of stacks over the formal neighborhood of a metric $[g] \in [\mathscr{M/D}]$. We then compute the function ring for this pullback over a formal stack and show that it gives us a ring of equivariant classical observables, as defined in \cite{cosgwill2}: Proposition \ref{themainthm} is thus one of the primary results of this paper, in that it explicitly links the stacky geometry presented earlier with the usual factorization algebra framework of Noether's Theorem presented in Costello and Gwilliam's books. This perspective is a beautiful fusion of Emmy Noether's foundational work in both homological algebra and symmetries in physics: homological algebra allows us to put external symmetries (perturbations) and internal symmetries (isometries) on equal footing, so that we can state a more fully encompassing form of Noether's Theorem.

In Remark \ref{equivariantfunctional} and Section \ref{stressenergymomentum}, we consider the conservation of the stress-energy tensor $T_{\mu\nu}$, the conserved quantity associated to general covariance via Noether's Theorem, in derived deformation theoretic terms: essentially, $T_{\mu\nu}$ tells us how the aforementioned formal stack acts on the field theories it is coupled to, in the language of $L_{\infty}$ algebras. In Theorem \ref{akeythm}, we expound on the above by computing a perturbative equivalence of observables when the theory is deformed by a vector field, and make a few remarks on how this might be relevant at higher orders in perturbation theory in Appendix \ref{higherorders}. One of the objectives of Section \ref{stressenergymomentum} is to provide a potential pathway for physicists to link their tools with ours.

\subsection{Future Directions} Much of this paper serves as a set-up for a few distinct projects. My primary motivation looking forward is the subject of anomalies in perturbative quantum field theory. Anomalies arise when the quantization of families of field theories over a parameter space with some classical symmetry does not necessarily respect that symmetry. In \cite{rabin19}, Rabinovich computes the BV quantization of families of free fermions parameterized by gauge fields and by this process recovers the axial anomaly. The anomaly is then explicitly quantified cohomologically by viewing the background gauge fields perturbatively (much in the way we consider metrics modulo diffeomorphisms perturbatively in Section \ref{the main result}): computations from BV quantization allow Rabinovich to equate the anomaly with the index of the original Dirac operator, as per the Atiyah-Singer Families Index Theorem. 

Much of the work in the current paper is motivated by the desire to reproduce similar computations to those in \cite{rabin19} when replacing connections modulo gauge by metrics modulo diffeomorphism. Before diving into \textit{quantization}, we must first understand both the global and perturbative nature of the stack $[\mathscr{M/D}]$ of metrics modulo diffeomorphism as it parameterizes \textit{classical} theories. In the case of free fermions parameterized by $[\mathscr{M/D}]$, we hope to reproduce a version of results stated in \cite{rabin20}: there, Rabinovich connects his work in \cite{rabin19} to defining a determinant line bundle $\mathrm{Det}(D)$ (\`{a} la Quillen) over some parameter space $B$ via BV quantization. In our case, we would let $B = [\mathscr{M/D}]$ and then the anomaly would constitute the first Chern class of the determinant line bundle over this.

Another goal is to understand how Wald's results on viewing black hole entropy as Noether charge (as in \cite{waldBH}) might be feasible within the BV framework. The study and physical interpretation of black hole thermodynamics is as popular now as it has ever been; however, it remains elusive in many regards. Wald's work connecting black hole entropy to Noether's Theorem may well serve as a point of connection to our work: in particular, the BV version of Noether's Theorem is put into detail and application in the second half of \cite{cosgwill2} and explicit computations in the case of metrics modulo diffeomorphism (a central object in Wald's paper) are provided in this article. It is also advantageous that Wald focuses on structural and algebraic aspects, so that porting it all over into the BV framework might be somewhat natural.

\section{Bundles of Batalin-Vilkovisky field theories}

We will begin by introducing the Batalin-Vilkovisky (BV) formalism: the purpose of the following narrative is to show how classical field theory is very naturally expressed in this formalism, especially in the context of diffeomorphism equivariance. The basic ingredients required from the outset are a space of fields, which define the \textit{kinematics} of a physical model, and an action functional, which fixes the \textit{dynamics} of that model. The fields on a space (or spacetime) $X$ are sections of some bundle $F \to X$ (usually a vector bundle), denoted $\mathscr{F} := \Gamma(X, F)$. The action functional is a function $S : \mathscr{F} \to \mathbf{R}$ whose critical locus $\mathrm{Crit}(S)$\footnote{This is computed via variational calculus, and described for example in Appendix E of \cite{wald}.} is the set of $\phi \in \mathscr{F}$ that satisfy the Euler Lagrange equations associated to $S$ via functional differentiation. To truly begin a discussion of the BV formalism, we must begin by making precise the notion of a functional.

\begin{rmk}
	So far, our bundle $F \to X$ is not graded: as we unfold what it means to define a BV theory, $F$ will be replaced by a differential graded bundle, but the notation will not change.
\end{rmk}

	\subsection{Introduction to the BV Formalism}\label{BVtheory}
	
The space of fields is usually infinite dimensional, which means we cannot take the usual algebraic symmetric powers of $\mathscr{F}$ to define their space of functions. Thus, we have the following definitions which play an identical role, but for the infinite dimensional case. Much of what follows is from Chapter 5, Section 3 of \cite{costello} and Chapter 3, Section 5 and Appendix B of \cite{cosgwill1}.

	\begin{defn}(Defined in Section 3.5.7 of \cite{cosgwill1})\label{functionals}
		The algebra of \textbf{functionals} on $\mathscr{F}$ is
		$$
		\mathscr{O}(\mathscr{F}) := \prod_{k \geq 0} \mathrm{Hom}(\mathscr{F}^{\otimes k}, \mathbf{R})_{S_{k}}.
		$$
	\end{defn}

We may sometimes denote this ring as $\mathrm{Sym}(\mathscr{F}^{\vee})$.
	
	\begin{rmk}
	To be fully precise, if $X$ is compact, $\mathscr{F} = \Gamma(X, F)$ is a nuclear Fr\'{e}chet space, where $\otimes$ denotes the completed projective tensor product, so that
	$$
	\mathscr{F}^{\otimes k} := \Gamma(X \times \cdots \times X, F \boxtimes \cdots \boxtimes F),
	$$
	meaning each $\mathrm{Hom}(\mathscr{F}^{\otimes k}, \mathbf{R})_{S_{k}}$ is a space of continuous multilinear functionals endowed with the strong dual topology: i.e. a space of distributions. The literature mentioned above defines all of this for a slightly broader class of spaces than Fr\'{e}chet spaces, but that is enough for us. In particular, we have the following fact, from page 1 of \cite{tuschmannwraith}:
	
	\begin{ex}\label{sectionsfrech}
Let $X$ be a smooth, compact, finite dimensional manifold, and let $F \to X$	be a vector bundle with space of sections $\Gamma(X, F) =: \mathscr{F}$. Choose Riemannian metrics and connections on $TX$ and $F$, let $\nabla^{i}\phi$ denote the $i^{th}$ covariant derivative of $\phi \in \mathscr{F}$, and set
$$
|| f ||_{n} := \sum_{i=0}^{n} \mathrm{sup} |\nabla^{i}\phi(x)|.
$$
By means of the topology defined by the sequence of norms $\{||-||_{n} \}$, $\mathscr{F}$ is a Fr\'{e}chet space. Clearly, we can define differential graded Fr\'{e}chet spaces as well, as will soon be relevant.
\end{ex}
	
	\end{rmk}

	\begin{defn}\label{locals}
	The \textbf{space of local functionals}, denoted $\mathscr{O}_{\mathrm{loc}}(\mathscr{F})$, is the linear subspace of $\mathscr{O}(\mathscr{F})$ spanned by elements of the form
$$
F_{k}(\phi) = \int_{X}(D_{1}\phi)(D_{2}\phi)\ldots(D_{k}\phi)\textrm{vol},
$$
for fields $\phi \in \mathscr{F}$ and differential operators $D_{i}$ on $X$.  
	\end{defn}

\begin{lem}(\cite{costello}, Ch. 5, Lemma 6.6.1)
There is an isomorphism of cochain complexes
	$$
	\mathscr{O}_{\mathrm{loc}}(\mathscr{F}) \cong \mathrm{Dens}_{X} \otimes_{D_{X}} \mathscr{O}_{\mathrm{red}}(\mathscr{J}(F)),
	$$	
where $\mathscr{J}(F)$ denotes sections of the $\infty$-jet bundle $\mathrm{Jet}(F) \to X$, and $\mathscr{O}_{\mathrm{red}}(\mathscr{J}(F))$ is the quotient of $\mathscr{O}(\mathscr{J}(F)) = \mathrm{Sym}(\mathscr{J}(F)^{\vee})$ by the constant polynomial functions. 
\end{lem}

\begin{rmk}
Sections of $\mathscr{O}_{\mathrm{loc}}(\mathscr{F})$ are exactly elements of the preceding form, and integration defines a natural inclusion:
$$
\iota : \mathscr{O}_{\mathrm{loc}}(\mathscr{F}) \to \mathscr{O}_{\mathrm{red}}(\mathscr{F}). 
$$
This lemma shows that $\mathscr{O}_{\mathrm{loc}}(\mathscr{F})$ is the space of Lagrangian densities modulo total derivatives: this is desirable because adding a total derivative to a Lagrangian density does not affect the dynamics described in the equations of motion. Local functionals are also more manageable in terms of functional analysis; for example, the action functional $S$ is always an element of $\mathscr{O}_{\mathrm{loc}}(\mathscr{F})$, and local functionals are key in defining the Poisson bracket, as we will see below.
\end{rmk}

\begin{defn}
	For $F \to X$ a graded vector bundle, a constant coefficient \textbf{$k$-shifted symplectic structure} is an isomorphism 
	$$
	F \cong_{\omega} F^{!}[k] := (\mathrm{Dens}_{X} \otimes F^{\vee})[k]
	$$
	of graded vector spaces that is graded antisymmetric. 
\end{defn}

\begin{rmk}

It stands to reason that a symplectic structure on a space defines a Poisson bracket on its space of functions: this is indeed the case for $\mathscr{O}_{\mathrm{loc}}(\mathscr{F}) \subset \mathscr{O}(\mathscr{F})$. This is \textit{not} the case however for all of  $\mathscr{O}(\mathscr{F})$, for functional analytic reasons which are outside the scope of this paper.\footnote{Details about this can be found in Chapter 4 of \cite{cosgwill2}.} We will denote the Poisson (anti-)bracket induced by $\omega$ as $\{-,-\}$.
	
\end{rmk}
	
\begin{defn}\label{BVdefn}
A \textbf{Batalin-Vilkovisky classical field theory} $(\mathscr{F}, \omega, S)$ on a smooth manifold $X$ is a differential $\mathbf{Z}$-graded vector bundle $F \to X$ equipped with a $-1$-shifted symplectic structure $\omega$ and an action functional $S \in \mathscr{O}_{\mathrm{loc}}(\mathscr{F})$ such that: \\
(1) $S$ satisfies the \textbf{classical master equation} (CME): $\{S,S\} = 0$. \\
(2) $S$ is at least quadratic, so that it can be written uniquely as $S(\varphi) = \omega(\varphi, Q\varphi) + I(\varphi)$, where $Q$ is a linear differential operator and $I \in \mathscr{O}_{\mathrm{loc}}(\mathscr{F})$ is at least cubic.  \\
A \textbf{free theory} is one in which $I=0$: i.e. the action functional $S$ is purely quadratic.
\end{defn}

\begin{rmk}
Although $\{-,-\}$ is not a Poisson bracket on $\mathscr{O}(\mathscr{F})$, bracketing with a local functional like $S \in \mathscr{O}_{\mathrm{loc}}(\mathscr{F})$ defines a derivation 
$$
\{S,- \} : \mathscr{O}(\mathscr{F}) \to \mathscr{O}(\mathscr{F})[1]
$$
regardless of whether or not the BV theory is free. For a \textit{free} theory, it can be shown that $\{S,- \} = Q$ on $\mathscr{O}(\mathscr{F})$, where the differential $Q$ on $\mathscr{F}$ is extended to $\mathscr{O}(\mathscr{F})$ as a derivation. For an \textit{interacting} theory, $\{S,-\}$ is prescribed by an $L_{\infty}$ algebra structure on $\mathscr{F}$, which we will describe in Definition \ref{Linftyalg} and provide examples of thereafter. The ellipticity or hyperbolicity of $(\mathscr{F}, Q)$ is sometimes assumed: this will be commented on later.
\end{rmk}

\begin{defn}
		Let $F \to X$ be a differential graded vector bundle with differential $Q$ on its sheaf of sections $\mathscr{F}$. Then the dg commutative ring of \textbf{(global) classical observables} for the theory defined by $(\mathscr{F}, Q)$ is 
		$$
		\mathrm{Obs}^{\mathrm{cl}}(X, \mathscr{F}) := (\mathscr{O}(\mathscr{F}), \{S,-\}).
		$$
	\end{defn}
	
\begin{rmk}\label{derivedcriticallocus}
	We will briefly describe how $\mathrm{Obs}^{\mathrm{cl}}(X, \mathscr{F})$ is to be understood as the dg ring of functions on the derived critical locus of the action functional $S : \mathscr{F}^{0} \to \mathbf{R}$, where the degree zero part of the dg fields is the ``na\"{i}ve" original space of fields, following \cite{cosgwill2}. 
	
	The ordinary critical locus $\mathrm{Crit}(S)$ is the intersection of the graph 
	$$
	\Gamma(dS) \subset T^{\vee}\mathscr{F}^{0}
	$$
	with the zero section $\mathscr{F}^{0} \subset T^{\vee}\mathscr{F}^{0}$. We thus get its commutative algebra of functions to be 
	$$
	\mathscr{O}(\mathrm{Crit}(S)) = \mathscr{O}(	\Gamma(dS)) \otimes_{\mathscr{O}(T^{\vee}\mathscr{F}^{0})} \mathscr{O}(\mathscr{F}^{0}).
	$$
	However, $\mathrm{Crit}(S)$ can be singular (e.g. it may be a non-transverse intersection), so we follow the philosophy of derived (algebraic) geometry and replace the above critical locus with the derived critical locus $\mathrm{Crit}^{h}(S)$, which has ring of functions
	$$
		\mathscr{O}(\mathrm{Crit}^{h}(S)) = \mathscr{O}(	\Gamma(dS)) \otimes_{\mathscr{O}(T^{\vee}\mathscr{F}^{0})}^{\mathbb{L}} \mathscr{O}(\mathscr{F}^{0}).
	$$
	This is now a commutative dg algebra instead of an ordinary commutative algebra, and it can be realized as the complex
	$$
	\mathscr{O}(T^{\vee}[-1]\mathscr{F}^{0}) = \Gamma(\mathscr{F}^{0}, \Lambda^{\bullet}T\mathscr{F}^{0}).
	$$
	Now we see how the Batalin-Vilkovisky dg fields $(\mathscr{F}, Q)$ arise naturally from a derived geometric perspective as $T^{\vee}[-1]\mathscr{F}^{0}$; moreover, the differential on $\Gamma(\mathscr{F}^{0}, \Lambda^{\bullet}T\mathscr{F}^{0})$ is contracting with the 1-form $dS \in \Omega^{1}(\mathscr{F}^{0})$, and this can be shown to be equivalent to $\{S,-\}$, as we would expect. 
	\end{rmk}

\begin{ex}\label{freescalarfield}
	A running example through much of this text will be scalar field theory. We will consider the \textit{free} case first. Fix a semi-Riemannian manifold $(X,g)$ and consider its space of smooth functions $\Gamma(X, \underline{\mathbf{R}}) = C^{\infty}(X)$: these are the a priori fields. The action functional is 
	\begin{equation}\label{scalfunctional}
	S_{g}(\varphi) = \frac{-1}{2}\int_{X} \varphi \Delta_{g}\varphi \mathrm{vol}_{g},
	\end{equation}
	where $\varphi \in C^{\infty}(X)$, $\mathrm{vol}_{g}$ is the volume form associated to the metric $g$, written in coordinates as $\sqrt{\det g} dx_{1} \wedge \ldots \wedge dx_{n}$, and the Laplace-Beltrami operator $\Delta_{g}$ associated to $g$ should not be mistaken for the BV Laplacian discussed in related literature. The Euler-Lagrange equation here is Laplace's equation, $\Delta_{g}\varphi = 0$, so that  $\mathrm{Crit}(S)$ is the set of harmonic functions. 
	
By the above, the derived critical locus is then
	\begin{equation}\label{BVfreescalarfield}
		\mathscr{F}_{g} = C^{\infty}(X) \xrightarrow{Q_{g}} \mathrm{Dens}(X)[-1],
	\end{equation}
	where $\mathrm{Dens}(X)$ is the appropriate dual to $C^{\infty}(X)$ and $Q_{g}\varphi = \Delta_{g}\varphi \mathrm{vol}_{g}$ is the differential, which imposes the Euler-Lagrange equations: it is written so as to take values in $\mathrm{Dens}(X)$ but also to capture all of the dependence on $g \in \mathrm{Met}(X)$ in the action functional. The symplectic structure $\omega$ on $\mathscr{F}_{g} = C^{\infty}(X) \xrightarrow{Q_{g}} \textrm{Dens}(X)[-1]$ is 
	$$
	\omega(\varphi, \mu) = \int_{X}  \varphi \mu,
	$$
	for $\varphi$ and $\mu$ in degrees 0 and 1, respectively. We can thus write $S_{g}(\varphi)$ as $\omega(\varphi, Q_{g}\varphi)$. 
	
	For $\mathscr{F}_{g} = C^{\infty}(X) \xrightarrow{Q_{g}} \mathrm{Dens}(X)[-1]$, the underlying graded ring of $\mathrm{Obs}^{\mathrm{cl}}(X, \mathscr{F}_{g})$ is $\mathscr{O}(\mathscr{F}_{g})$, so that it is concentrated in nonpositive degrees, as Definition \ref{functionals} implies. The action functional $S_{g}(\varphi)$ defined in Equation (\ref{scalfunctional}) is a degree 0 element of $\mathscr{O}(\mathscr{F}_{g})$, but also defines a degree 1 differential on $\mathscr{O}(\mathscr{F}_{g})$ as $\{S_{g}, - \}$: thus, $\{S_{g}, S_{g}\}$ must be a degree 1 element of $\mathscr{O}(\mathscr{F}_{g})$. Since in this example $\mathscr{O}(\mathscr{F}_{g})$ is concentrated in nonpositive degrees, the classical master equation holds vacuously. Thus, the free massless scalar field with metric background $g$ defines a free BV classical field theory, since the other requirements are easily satisfied.

\begin{rmk}
It is implied here that $g$ is a Riemannian metric, because the associated partial differential operator is the elliptic Laplace-Beltrami operator. If $g$ were Lorentzian, then we would instead have the hyperbolic d'Alembertian $\Box_{g}$. For further details comparing these two regimes for the free scalar field, one should consult the thorough reference \cite{gwillrejz}.
\end{rmk}

\begin{rmk}\label{bundleoffields}
An advantage to shifting from the ``ordinary" fields $C^{\infty}(X)$ to the derived critical locus $\mathscr{F}_{g}$ is that there now is an \textit{explicit dependence} in the fields on the metric $g \in \mathrm{Met}(X) =: \mathscr{M}$.\footnote{We will denote $\mathrm{Met}(X)$ as $\mathscr{M}$ when $X$ is implicit.} This will allow us to define a \textit{differential graded vector bundle} $\pi : \mathscr{F} \to \mathscr{M}$: the base space is the space of all (semi-)Riemannian metrics on $X$ and the fibers $\pi^{-1}(g) = \mathscr{F}_{g}$ are field theories depending on the fixed $g$. This opens up the possibility of seeing how varying the background metric effects the field theory. We have such a dg vector bundle only when the theory is free (i.e. $S$ is quadratic in $\varphi$): for an interacting theory, we will require the notion of an $L_{\infty}$ algebra and bundles thereof.
\end{rmk}
\end{ex}

\begin{defn}\label{Linftyalg}
An \textbf{$L_{\infty}$ algebra} over $R$ is a $\mathbf{Z}$-graded, projective $R$-module $\mathfrak{g}$ with a sequence of multilinear maps of cohomological degree $2-n$: 
$$
\ell_{n} : \mathfrak{g} \otimes_{R} \ldots \otimes_{R} \mathfrak{g} \to \mathfrak{g},
$$
where $n \in \mathbf{N}$, such that all $\ell_{n}$ are (1) graded antisymmetric and (2) satisfy the $n$-Jacobi rule.\footnote{We are sweeping details for this rule under the rug: Definition A.1.2 in \cite{cosgwill2} is the whole megillah.}
\end{defn}

\begin{ex}\label{phi4}
	The most natural example of an $L_{\infty}$ algebra for us comes from encoding nonlinear partial differential equations: i.e. those associated to an \textit{interacting} field theory, with a degree three or higher action functional. 
	
	For example, say we want to encode $\Delta_{g}\varphi + \frac{1}{3!}\varphi^{3} = 0$, the Euler Lagrange equation associated to the action functional
	$$
	S_{g}(\varphi) = \frac{-1}{2}\int_{X}\varphi \Delta_{g}\varphi \mathrm{vol}_{g} + \frac{1}{4!}\int_{X}\varphi^{4} \mathrm{vol}_{g}.
	$$
	The pertinent $L_{\infty}$ algebra has underlying cochain complex
	$$
	\mathscr{L} = C^{\infty}(X)[-1] \to \mathrm{Dens}(X)[-2],\footnote{Details explaining the necessity of this shift up by 1 can be found in 4.2 of \cite{cosgwill2}.}
	$$
	where the differential is $Q_{g}\varphi = \Delta_{g}\varphi \mathrm{vol}_{g}$ and the only higher bracket is $\ell_{3} : C^{\infty}(X)^{\otimes 3} \to \mathrm{Dens}(X)$, defined as $\ell_{3} : \varphi_{1} \otimes \varphi_{2} \otimes \varphi_{3} \mapsto \varphi_{1}\varphi_{2}\varphi_{3}\mathrm{vol}_{g}$. Letting $(R, \mathfrak{m}_{R})$ be a nilpotent Artinian ring in degree $0$, we get that $\varphi \in C^{\infty}(X) \otimes \mathfrak{m}_{R}$ satisfies the Maurer-Cartan equation $\mathscr{L}$ if and only if 
	$$
	Q_{g}\varphi + \frac{1}{3!}\varphi^{3}\mathrm{vol}_{g} = 0,
	$$
which recovers the desired partial differential equation (with values in densities).   Thus we see how an $L_{\infty}$ algebra quantifies how a given equation fails to be linear (a free theory has only nontrivial $\ell_{1}$, and so only requires a dg structure to be described). Moreover, $\mathscr{L}$ is an even more particular object, which we now define.
\end{ex}

\begin{defn}
A \textbf{local $L_{\infty}$ algebra} on a manifold $X$ is: \\
(1) A graded vector bundle $L \to X$, where we denote the sections as $\mathscr{L}$, \\
(2) a differential operator $d : \mathscr{L} \to \mathscr{L}$ of cohomological degree 1 such that $d^{2} = 0$, and \\
(3) a collection of polydifferential operators $\ell_{n} : \mathscr{L}^{\otimes n} \to \mathscr{L}$ for $n \geq 2$ which are alternating, of cohomological degree $2-n$, and which make $\mathscr{L}$ an $L_{\infty}$ algebra. \\
If the local $L_{\infty}$ algebra $(\mathscr{L},d)$ is an elliptic complex, we call it an \textbf{elliptic $L_{\infty}$ algebra}.
\end{defn}

The $L_{\infty}$ algebra $\mathscr{L}$ of our ongoing example is an elliptic $L_{\infty}$ algebra. One advantage of introducing this notion is the following definition of observables for perturbative field theories, which we will employ in Section \ref{the main result} when discussing formal computations:

\begin{defn}[Definition 5.1.1 in \cite{cosgwill2}]\label{obsonaset}
	The \textbf{observables with support in the open subset $U$} are the commutative dg algebra 
	\begin{equation}
	\mathrm{Obs}^{\mathrm{cl}}(U) := C^{\bullet}(\mathscr{L}(U)),	
	\end{equation}
	where $C^{\bullet}(\mathscr{L})$ denotes Chevalley-Eilenberg cochains. The \textbf{factorization algebra of observables} for this classical field theory, denoted $\mathrm{Obs^{cl}}$, assigns $\mathrm{Obs}^{\mathrm{cl}}(U)$ to each open $U \subset X$. 
\end{defn}

\begin{rmk}
The computations in Example \ref{phi4} work just fine if we replace the elliptic operator $\Delta_{g}$ with the hyperbolic wave operator $\Box_{g}$, so it would be convenient to specify the dynamics of the Lorentzian analogue of $\varphi^{4}$ theory with an $L_{\infty}$ algebra, too. However, as commentary in Gwilliam and Rejzner's paper \cite{gwillrejz} suggests, comparisons between the Lorentzian and Riemannian settings get stickier when considering interacting theories. A precise definition of the correct notion of a hyperbolic  $L_{\infty}$ algebra is presented in the recent paper \cite{bms}.
\end{rmk}

\begin{rmk}\label{scalarcovariance}
	In Example \ref{freescalarfield} of the free scalar field, both components of the graded space of fields have an action by the diffeomorphism group of $X$--denoted $\mathscr{D}$ when unambiguous--via pullback: for $f \in \mathscr{D}, \varphi \in C^{\infty}(X),$ and $\mu \in \textrm{Dens}(X)$, $f \cdot \varphi = f^{*}\varphi = \varphi \circ f$ and $f \cdot \mu = f^{*}\mu$. Additionally, $\mathscr{D}$ acts on $\textrm{Met}(X)$ via pullback: $f \cdot g = f^{*}g$. What is special about this example is that the differential $Q_{g}$ commutes with diffeomorphisms in the following sense: $f^{*}(Q_{g}\varphi) = f^{*}(\Delta_{g}\varphi \textrm{vol}_{g}) = \Delta_{f^{*}g} (f^{*}\varphi) \textrm{vol}_{f^{*}g} = Q_{f^{*}g}(f^{*}\varphi)$. This result is equivalent to the fact that the Laplacian commutes with diffeomorphisms\footnote{This is computed by expressing $\Delta_{g} = \mathrm{div}_{g}\mathrm{grad}$, and is done explicitly in notes by Y. Canzani, available at: \url{https://www.math.mcgill.ca/toth/spectral geometry.pdf}.}. This suggests that if we parameterize families of free scalar BV theories by $\textrm{Met}(X)$, the result will be a ``$\mathscr{D}$-equivariant bundle". In fact, we can show how this can work for interacting theories: but first, we must make precise the idea of a differential graded equivariant bundle.
\end{rmk}

\subsection{Equivariant Vector Bundles}

To discuss general covariance, we must first understand what an equivariant vector bundle is and how to use one to specify a family of field theories parameterized by semi-Riemannian metrics. Once that is done and we make the connection with general covariance, we will see how groupoids and stacks naturally arise in this context and provide additional advantages.

\begin{defn}[Definition 1.5 of \cite{bgv}]\label{equivvect}
Let $G$ be a Lie group. A smooth fiber bundle $\pi : E \to M$ is said to be \textbf{$G$-equivariant} if: (i) both $E$ and $M$ are left $G$-spaces, and (ii) $\pi : E \to M$ is a $G$-equivariant map. If $E=V$ is a vector bundle, we also require that for all $g \in G$ and $p \in M$, $g : V_{p} \to V_{g \cdot p}$ is a linear transformation, where $V_{p} := \pi^{-1}(p)$. 
\end{defn}

One must be mindful that within the vector bundle part of this definition is packaged the information that the \textit{fibers} $V_{p}$ of $V \to M$ could themselves be $G$-spaces over fixed point sets; however, if there are no fixed points of $G$ acting on $M$, we immediately get the following.

\begin{thm}\label{vectequiv}
For $M$ a smooth $G$-space on which the Lie group $G$ acts freely, there is an equivalence of categories between vector bundles on $M/G$ and $G$-equivariant vector bundles:
\begin{equation}
\textup{VectBun}(M/G) \xrightarrow{\cong} \textup{VectBun}_{G}(M).
\end{equation}
\end{thm}

\begin{rmk}
This theorem nicely encapsulates how we might keep track of linear data parameterized by an underlying space which has some symmetries: we simply quotient the underlying space out by its symmetries and look at the vector bundle over that. However, this is where the problem becomes apparent: if there are any points in $M$ which are stabilized by $G$ or any of its nontrivial subgroups, then $M/G$ is no longer a smooth manifold at those points. This makes it more difficult to associate to it any structures \textit{which depended on the differentiability of} $M$, like its ring of smooth functions or sections of certain bundles. Stacks deal with those issues nicely, and provide an analogous theorem in the case the action $M \act G$ is not free.
\end{rmk}

There is one more distinction which is significant for this work, which we will define here: the fibers of the vector bundles we want to consider are differential graded as well as equivariant.

\begin{defn}
	A \textbf{differential graded vector bundle} is a vector bundle $V \to M$ whose fibers $V_{p}$ are $\mathbf{Z}$-graded vector spaces with a smoothly varying differential  $Q_{p}^{i} : V_{p}^{i} \to V_{p}^{i+1}$. 
\end{defn}

We will usually abbreviate ``differential graded" as ``dg", and such a vector bundle may sometimes be denoted $V^{\bullet} \to M$ or $(V^{\bullet} \to M, Q)$, depending on what we would like to emphasize within certain contexts. We have a similar definition when the fibers are $L_{\infty}$ algebras:

\begin{defn}
		A \textbf{bundle of (elliptic) $L_{\infty}$ algebras} is a $\mathbf{Z}$-graded vector bundle $\pi : (V, \ell) \to M$\footnote{We may sometimes omit this notation as a pair if the $L_{\infty}$ structure is implicit.} whose fibers $(V_{p}, \ell^{p}) := \pi^{-1}(p)$ are (elliptic) $L_{\infty}$ algebras, such that the $L_{\infty}$ structure varies smoothly over $M$.
		
	\end{defn}

\begin{defn}\label{dgequiv}
	A dg vector bundle $(V^{\bullet} \to M, Q)$ is \textbf{$G$-equivariant} if: (i) each of the $V^{i} \to M$ is $G$-equivariant in the usual sense and (ii) the action by $G$ induces a cochain map between fibers, i.e. for $g \in G$ and for $i \in \mathbf{Z}$, the following square commutes:
	
	\begin{center}
\begin{tikzcd}
V^{i}_{p} \arrow[r, "Q_{p}^{i}"] \arrow[d, "g \cdot"]
& V^{i+1}_{p}  \arrow[d, "g \cdot"] \\
V^{i}_{g \cdot p}  \arrow[r, "Q_{g \cdot p}^{i}"]
& V^{i+1}_{g \cdot p}.
\end{tikzcd}
\end{center}
\end{defn}
A totally analogous definition holds for $G$-equivariant bundles of $L_{\infty}$ algebras. 

\begin{rmk}

It might be the case that the differential (or $L_{\infty}$ structure) does not depend on $p \in M$ for trivial $V \to M$: in this case, the bundle is still $G$-equivariant, in a rather trivial way. However, it is easy to find examples in which there is such a dependence, as this the case for our version of general covariance. 

\end{rmk}

\subsection{General Covariance}\label{GCsection}

To state general covariance rigorously in our sense, we must first introduce a few facts about Fr\'{e}chet manifolds. The space of metrics on a smooth manifold and its group of diffeomorphisms are infinite dimensional, so defining vector bundles or other structures which depend on the differentiability of $\textrm{Met}(X)$ will require us to consider a special class of manifolds called \textit{Fr\'{e}chet manifolds}, which we now define (much of what we state is adapted from  \cite{tuschmannwraith}\footnote{Another helpful reference is \url{https://ncatlab.org/nlab/show/Fr\%C3\%A9chet+manifold}.}).

\begin{defn}[Adapted from Definition 1.3 in \cite{tuschmannwraith}]
A \textbf{Fr\'{e}chet manifold} is a Hausdorff topological space with an atlas of coordinate charts taking values in Fr\'{e}chet spaces (i.e. complete, Hausdorff, metrizable, locally convex vector spaces) such that the transition functions are smooth maps between Fr\'{e}chet spaces. 
\end{defn}

\begin{ex}
$\Gamma(X, \mathrm{Sym}^{2}(T^{\vee}_{X}))$, of which $\mathrm{Met}(X)$ is an open submanifold, is a Fr\'{e}chet manifold if $X$ is compact, and similarly, the diffeomorphism group $\mathrm{Diff}(X)$ is a Fr\'{e}chet Lie group as long as $X$ is compact \cite{tuschmannwraith}. Thus, we will usually assume that $X$ is compact or even closed in much of what follows, even though in the Lorentzian case, $X$ is usually not compact. However, many physically relevant Lorentzian manifolds are assumed to have the simple form $\Sigma \times \mathbf{R}$, for $\Sigma$ a spacelike compact submanifold and $\mathbf{R}$ the time direction. This is the path through which many Riemannian results are translated into the Lorentzian regime. 
\end{ex}

\begin{defn}[Adapted from Definition 1.7 of \cite{tuschmannwraith}]
	The space $\mathrm{Met}(X) = \mathscr{M}$ of all Riemannian metrics on a compact $X$ is the Fr\'{e}chet manifold defined as the subspace of $\Gamma(X, \mathrm{Sym}^{2}(T^{\vee}_{X}))$ consisting of all sections which are Riemannian metrics on $X$, equipped with the smooth topology of uniform convergence on compact subsets.
\end{defn}

Since $\mathscr{M}$ is a Fr\'{e}chet manifold and since any space of fields $\mathscr{F}$ on a compact $X$ for a BV classical field theory is a dg Fr\'{e}chet manifold by means of being a dg Fr\'{e}chet space (by Definition \ref{BVdefn} and Example \ref{sectionsfrech}), we can bring to fruition Remark \ref{bundleoffields}:

\begin{prop}\label{fieldsasfrechetbundle}
Any BV classical field theory for which the action functional $S$ depends on the metric $g \in \mathscr{M}$ defines a dg Fr\'{e}chet vector bundle $\pi :(\mathscr{F}, \{ S,- \}) \to \mathscr{M}$ for a free theory or a dg Fr\'{e}chet bundle of $L_{\infty}$ algebras for an interacting theory, with fibers $\pi^{-1}(g) = (\mathscr{F}_{g}, \{S_{g},-\})$. 
\end{prop}
\begin{proof}
	Because $\mathscr{M}$ is always contractible, the underlying graded vector bundle is $\mathscr{F} \times \mathscr{M}$, where $\mathscr{F}$ is Fr\'{e}chet by Example \ref{sectionsfrech}. A product of Fr\'{e}chet manifolds is once again Fr\'{e}chet, and the assignment of a dg or $L_{\infty}$ structure is smooth. 
\end{proof}

The computation in Remark \ref{scalarcovariance} along with Proposition \ref{fieldsasfrechetbundle} allow us to state the following:

\begin{lem}\label{scalequiv}
For the free scalar BV theory defined in Example \ref{freescalarfield}, any diffeomorphism $f \in \mathscr{D}$ defines a cochain map between fibers of the dg Fr\'{e}chet vector bundle $(\mathscr{F}, Q) \to \mathscr{M}$: 
\begin{center}
\begin{tikzcd}
\mathscr{F}_{g} = C^{\infty}(X) \arrow[r, "Q_{g}"] \arrow[d, "f^{*}"]
& \mathrm{Dens}(X)[-1] \arrow[d, "f^{*}"] \\
\mathscr{F}_{f^{*}g} = C^{\infty}(X) \arrow[r, "Q_{f^{*}g}"]
& \mathrm{Dens}(X)[-1],
\end{tikzcd}
\end{center}
which implies that $(\mathscr{F}, Q) \to \mathscr{M}$ is a $\mathscr{D}$-equivariant differential graded vector bundle.
\end{lem}

For the remainder of this article, we will sometimes drop the term ``Fr\'{e}chet" when it is contextually implied, unless attention is otherwise drawn to it. 
This result also implies a significant and useful corollary:

\begin{cor}\label{representations}
If $g \in \mathscr{M}$ is a fixed point of $f \in \mathscr{D}$ (i.e. if $f$ is an isometry of $g$) and if $Q_{g}\varphi = 0$, then $Q_{g}(f^{*}\varphi) = 0$. In other words, isometries of the metric $g$ act on the space of solutions to $\Delta_{g}\varphi = 0$ (Laplace's equation). 
\end{cor}

Of course, this corollary holds for \textit{any} generally covariant BV field theory: we bring special attention to it in this case because it is a ``gold standard" result when learning PDE for the first time, and thus serves as a touchstone for the value of the preceding perspective.

\begin{rmk}
The differential $Q_{g}$ of the differential graded scalar fields has a very clear dependence on the base space $\mathscr{M}$. In fact, as a topological space, the bundle is trivial, as it is $(C^{\infty}(X) \oplus \textrm{Dens}(X)[-1]) \times \mathscr{M}$: the differential $Q_{g}$ defines any nontriviality as a \textit{differential graded} vector bundle.
\end{rmk}

\begin{ex}\label{pertcon1}
Lemma \ref{scalequiv} holds for a particular case in which the BV theory is both \textit{free} and \textit{non-perturbative}: i.e. the Euler-Lagrange equations are linear in the fields $\phi \in \mathscr{F}_{g}$ and we are \textit{not} choosing a fixed solution to perturb around, so that the observables are polynomial functions of the fields as opposed to Taylor series.  

We will now consider an example of an interacting theory. The bundle $(\mathscr{L}, \{S,- \} ) \to \mathscr{M}$\footnote{Note that the notation has changed since the perturbative space of fields is $\mathscr{L} = \mathscr{F}[-1]$.} representing the family of theories is no longer just a dg vector bundle, but a bundle of elliptic $L_{\infty}$ algebras over $\mathscr{M}$. Heuristically speaking, we will no longer view the family as a collection of vector spaces varying over $\mathscr{M}$, but rather as a collection of formal neighborhoods varying over $\mathscr{M}$: although the underlying graded structure is still a vector bundle, the geometry encoded in the $L_{\infty}$ structures on distinct fibers implies this shift in perspective.

Let us return to Example \ref{phi4}. Recall that the equation of motion in that instance is:
$$
Q_{g}\varphi + \frac{1}{3!}\varphi^{3}\mathrm{vol}_{g} = 0. 	
$$
If we fix a diffeomorphism $f \in \mathscr{D}$, we see that the Euler-Lagrange form satisfies:
\begin{equation}\label{GCphi4}
	f^{*}(Q_{g}\varphi + \frac{1}{3!}\varphi^{3}\mathrm{vol}_{g}) = Q_{f^{*}g}(f^{*}\varphi) + \frac{1}{3!} (f^{*}\varphi)^{3}\mathrm{vol}_{f^{*}g}. 
\end{equation}
	The equivariance property for the first summand is precisely what is shown in Lemma \ref{scalequiv}, and the second summand (the interaction term) is equivariant because polynomial functions of the fields are patently equivariant in this way. Equation (\ref{GCphi4}) can then be reformulated in terms of the brackets on the elliptic $L_{\infty}$ algebra of Example \ref{phi4} as: 
\begin{equation}
	f^{*}\big( \ell_{1}^{g}(\varphi) + \frac{1}{3!}\ell_{3}^{g}(\varphi, \varphi, \varphi) \big) = \ell_{1}^{f^{*}g}(f^{*}\varphi) + \frac{1}{3!}\ell_{3}^{f^{*}g}(f^{*}\varphi, f^{*}\varphi, f^{*}\varphi),
\end{equation}
where we have included the dependence of the brackets $\ell_{k}$ on the underlying metric $g \in \mathscr{M}$ as a superscript. The above equation is the $\mathscr{D}$-equivariance property we desire in the Euler-Lagrange term which implies that the family of theories defined by $\varphi^{4}$ theory as in Example \ref{phi4} is generally covariant. 

This generalizes naturally to the case in which the interaction term is any polynomial in $\varphi$ times $\mathrm{vol}_{g}$. In that case, $\ell_{1} = Q_{g}$ and $\ell_{n} : C^{\infty}(X)[-1]^{\otimes n} \to \mathrm{Dens}(X)[-2]$ for $n \geq 2$ is:
$$
\ell_{n} : \varphi_{1} \otimes \ldots \otimes \varphi_{n} \mapsto  \lambda_{n} \varphi_{1}\ldots \varphi_{n}\mathrm{vol}_{g},
$$
where the $\lambda_{n}$ are constants. Similarly to Equation (\ref{GCphi4}), it is quick to show that:
\begin{equation}
f^{*} \big(\ell_{1}^{g}(\varphi) + \sum_{n \geq 2}\frac{\lambda_{n}}{n!} \ell_{n}^{g}(\varphi, \ldots, \varphi) \big)	 = \ell_{1}^{f^{*}g}(f^{*}\varphi) + \sum_{n \geq 2}\frac{\lambda_{n}}{n!} \ell_{n}^{f^{*}g}(f^{*}\varphi, \ldots, f^{*}\varphi). \footnote{The $n!$ is traditionally included to simplify computations with the $L_{\infty}$ structure.}
\end{equation}
Thus, any scalar field theory with action functional 
$$
S_{g}(\varphi) = \int_{X} (\frac{-1}{2}\varphi \Delta_{g}\varphi + V(\varphi))\mathrm{vol}_{g},
$$
where $V(\varphi)$ is a polynomial ``potential" in $\varphi$, is generally covariant. 
\end{ex}

\begin{lem}\label{polynomialinteraction}
Let $\pi : (\mathscr{L}, \{S,-\}) \to \mathscr{M}$ be a family of perturbative Batalin-Vilkovisky classical scalar field theories with polynomial potential. Any $f \in \mathscr{D}$ defines an $L_{\infty}$ map between fibers of $\pi : (\mathscr{L}, \{S,-\}) \to \mathscr{M}$: 
\begin{center}
\begin{tikzcd}
\mathscr{L}_{g} = C^{\infty}(X)[-1] \arrow[r, " \ell^{g} "] \arrow[d, "f^{*}"]
& \mathrm{Dens}(X)[-2] \arrow[d, "f^{*}"] \\
\mathscr{L}_{f^{*}g} = C^{\infty}(X)[-1] \arrow[r, " \ell^{f^{*}g} "]
& \mathrm{Dens}(X)[-2].
\end{tikzcd}
\end{center}
In other words, $\pi : (\mathscr{L}, \{S,-\}) \to \mathscr{M}$ is a $\mathscr{D}$-equivariant bundle of $L_{\infty}$ algebras. 
\end{lem}

An analogous version of Corollary \ref{representations} holds here, and follows by a nearly identical computation. We can now state a first definition for general covariance:

\begin{defn}\label{gencov}
	Let $\pi: (\mathscr{F}, \{S,-\}) \to \mathscr{M}$ define a family of BV field theories on $X$ parameterized by the space of metrics on $X$. If it is $\textrm{Diff}(X) =: \mathscr{D}$-equivariant as a differential graded vector bundle or as a bundle of $L_{\infty}$ algebras (depending on whether the theories are free or perturbative/interacting), we call the theory \textbf{generally covariant}. 
\end{defn}


\begin{rmk}
	Field theories which satisfy general covariance are therefore not sensitive to \textit{all} of $\mathscr{M}$, but only to the moduli space of metrics modulo diffeomorphism, $\mathscr{M}/\mathscr{D}$. Although many physically relevant metrics have many isometries, the coarse quotient $\mathscr{M}/\mathscr{D}$ ``forgets them": the need for a more general concept of a space which remembers them is desirable, and this is where stacks will become useful.
\end{rmk}

\begin{ex}
A tangible example of $\mathscr{M}/\mathscr{D}$ with such a singular point is the Riemannian manifold $X = \mathbf{R}^{n}$ along with the flat metric $\eta$. It is well known that the isometry group of $(\mathbf{R}^{n}, \eta)$	 is $O(n) \ltimes \mathbf{R}^{n}$, where the $\mathbf{R}^{n}$ in the semidirect product is the additive group of translations of $\mathbf{R}^{n}$. In particular, $O(n) \ltimes \mathbf{R}^{n}$ is a subgroup of $\mathrm{Diff}(\mathbf{R}^{n})$ which stabilizes $\eta \in \mathscr{M}(\mathbf{R}^{n})$, meaning that the corresponding point in the quotient is singular. Moreover, $O(n) \ltimes \mathbf{R}^{n}$ acts on the space of solutions to any generally covariant theory on $(\mathbf{R}^{n}, \eta)$	. Definition \ref{gencov} therefore ``enlarges" our usual idea of equivalence beyond isometries.
\end{ex}

\subsubsection{An important remark on functoriality}\label{functors}

	Strictly speaking, the claims of Lemmas \ref{scalequiv} and \ref{polynomialinteraction} are true in a broader sense. Instead of assuming that $f : X \to X$ is a diffeomorphism, we let $f : U \to X$ be an isometric embedding. In other words, consider the category $\mathbf{Riem}_{n}$ whose objects are Riemannian $n$-folds and whose morphisms are isometric embeddings: $f : (U,g') \to (X,g)$ so that $f^{*}g = g'$. In the case of the free scalar field in Lemma \ref{scalequiv}, the commutative square is replaced by
\begin{center}
	\begin{tikzcd}
	{C^{\infty}(X)} && {\mathrm{Dens}(X)[-1]} \\
	\\
	{C^{\infty}(U)} && {\mathrm{Dens}(U)[-1]}
	\arrow["{f^{*}}", from=1-1, to=3-1]
	\arrow["{Q_{g'} = Q_{f^{*}g}}", from=3-1, to=3-3]
	\arrow["{Q_{g}}", from=1-1, to=1-3]
	\arrow["{f^{*}}", from=1-3, to=3-3],
\end{tikzcd}
\end{center}
which commutes by the very same computation. This implies that the assignment of a free BV theory is a contravariant functor $ (\mathscr{F},Q) : \mathbf{Riem}_{n} \to \mathbf{dgVect}$ from the site of Riemannian $n$-folds to the category of cochain complexes. We can call this more general notion ``very general covariance" or keep it simply as ``general covariance". The computation from Lemma \ref{polynomialinteraction} implies that the above works out just as well for interacting theories: in that case, the target must be $L_{\infty}\mathbf{Alg}$, the category of $L_{\infty}$ algebras.\footnote{We will stick with the broader category of $L_{\infty}$ algebras for the rest of this section.} This suggests something deeper about the physics: not only are the computations invariant with respect to coordinates choices, but also ``manifold choices" more broadly.

We can compose the preceding functor with the functor $L_{\infty}\mathbf{Alg} \to \mathbf{dgAlg}$\footnote{As it stands, the target category can be $\mathbf{dgCAlg}$ (dg commutative algebras), but we leave it as is because we may lose commutativity after quantization.} which takes an $L_{\infty}$ algebra $\mathscr{L}$ and outputs its Chevalley-Eilenberg cochains $C^{\bullet}(\mathscr{L})$. Then the composite functor 
\begin{equation}\label{obsasfunctor}
\mathrm{Obs}^{\mathrm{cl}} : \mathbf{Riem}_{n} \to \mathbf{dgAlg}
\end{equation}
 is covariant, as indeed it should be if we would like to make a factorization algebra from it (as is done in \cite{cosgwill1} and \cite{cosgwill2}). This is a point of connection with the definition of covariance presented in \cite{fewster}. In that work, Fewster outlines a broad framework to understand the idea of ``same physics in all spacetimes" (SPAS) in which he defines (Definition 3.1) a locally covariant theory to be a covariant functor $\mathfrak{A} : \mathbf{BkGrnd} \to \mathbf{Phys}$ from some category of ``background geometries" to an appropriate category of ``physical quantities", like observables. Our preceding construction clearly falls into this class of objects. 
 
In the study of Algebraic Quantum Field Theory (AQFT), a common choice for $\mathfrak{A}$ is 
\begin{equation}
\mathfrak{A} : \mathbf{Loc}_{n} \to C^{*}\text{-}\mathbf{Alg}.
\end{equation}
$\mathbf{Loc}_{n}$ is the category of oriented, time-oriented, and globally hyperbolic $n$-dimensional Lorentzian manifolds whose morphisms are isometric embeddings which respect orientations and time-orientations, and $C^{*}$-$\mathbf{Alg}$ is the category of $C^{*}$ algebras, to which the observables of a quantum field theory (usually) belong. The work of this article pertains to classical observables, and much of the focus is on the first part of the composite functor $\mathrm{Obs^{cl}}$: but once the full composition is made, the comparison with AQFT is apparent. Further details concerning this subject are provided in great detail in \cite{rejzner} (in particular Section 2.5). 
 
\begin{rmk}
	To summarize, the contents of this paper are presented for a \textit{fixed} smooth manifold $X$, its space of metrics, its diffeomorphism group, and various fields defined on it because focusing on the ``smaller problem" made it easier to manage the functional analytic constructions presented earlier and invoked later on.
	
Because the aforementioned fields and groups are sheaves on $X$, it is already apparent that all of the work lifts to the level of the slice category $\mathbf{Riem}_{n}/X$,\footnote{This is a ``little site" built from the site $\mathbf{Riem}_{n}$. Note that once $X$ is fixed, $\mathrm{Diff}(X)$ acts on this site.} whose objects isometrically embed into $X$ and whose morphisms $f$ are specified by commutative triangles
	\begin{center}
		\begin{tikzcd}
	U && V \\
	& X.
	\arrow["f", hook, from=1-1, to=1-3]
	\arrow["{\iota_{V}}",  from=1-3, to=2-2]
	\arrow["{\iota_{U}}"',  from=1-1, to=2-2]
\end{tikzcd}
	\end{center}
	From here, it is not a stretch to see that our constructions lift to $\mathbf{Riem}_{n}$. In particular, this means we have the composition of functors
	\begin{equation}
		\mathbf{Riem}_{n} \xrightarrow{\mathscr{F}} L_{\infty}\mathbf{Alg} \xrightarrow{\mathrm{Obs^{cl}}} \mathbf{dgAlg}.\footnote{Letting the target category be $\mathbf{dgVect}$ is also an acceptable viewpoint.}
	\end{equation}
Even better, since we take for granted Costello and Gwilliam's result that $\mathrm{Obs^{cl}} : \mathbf{Disj}_{X} \to \mathbf{dgAlg}$ defines a factorization algebra for a fixed $X$, the above composition ultimately allows us to state the following:

\begin{prop}\label{functorial}
	Any generally covariant BV field theory $(\mathscr{F}, \omega, S)$ defines a functor
	\begin{equation}
		\mathrm{Obs^{cl}}(-, \mathscr{F}) : \mathbf{Riem}_{n} \to \mathbf{dgAlg}
	\end{equation}
	which constitutes a factorization algebra on the site $\mathbf{Riem}_{n}$. 
\end{prop}

\begin{rmk}
Concretely, once we input some $X \in \mathbf{Riem}_{n}$, the output is a factorization algebra on $X$. Roughly, a prefactorization algebra $\mathcal{F}$ on $X$ is a functor which takes disjoint opens $U_{i}$ as subsets of some larger open $V \subseteq X$ and outputs multiplication maps $\bigotimes_{i} \mathcal{F}(U_{i}) \to \mathcal{F}(V)$. A factorization algebra is a prefactorization algebra which satisfies a particular (co)descent axiom. Further details can be found in Sections 3.1 and 6.1 of \cite{cosgwill1}.
\end{rmk}

For the rest of the paper, we make constructions and compute results within the framework of stacks and formal stacks: one of the ultimate motivations is to step back and notice that all of the results hold at the level of generality specified in this subsection. An eventual goal is to make connections with other literature on functorial perspectives in field theory. An example of such literature linking AQFTs and factorization algebras is \cite{bps} (Theorem 4.7 in particular).

\end{rmk}


\subsection{Groupoids and Stacks}\label{groupoids}

Stacks provide a wonderful packaging of a quotient space, but before diving into them, we must quickly review \textit{groupoids}, which are the cornerstone of the theory of stacks and allow us to do concrete computations with them. For the most part, we follow the constructions outlined in \cite{carchedi} and \cite{heinloth}.

\begin{defn}
A \textbf{groupoid} $\mathcal{G}$ is a small category in which all arrows are invertible. Common notation is $\mathcal{G} = \mathcal{G}_{1} 
\doublerightarrow{s}{t} \mathcal{G}_{0}$, where $\mathcal{G}_{1}$ is the set of arrows and  $\mathcal{G}_{0}$ the set of objects; $s$ sends an arrow to its source object, and $t$ sends it to its target. Every such $\mathcal{G}$ has an identity map $e : \mathcal{G}_{0} \to \mathcal{G}_{1}$ sending an object to its identity arrow, an inverse map $i : \mathcal{G}_{1} \to \mathcal{G}_{1}$ sending an arrow to its inverse, and a multiplication map $m : \mathcal{G}_{1} \times_{\mathcal{G}_{0}} \mathcal{G}_{1} \to \mathcal{G}_{1}$ that concatenates arrows. $s, t, e, i,$ and $m$ are called the \textbf{structure maps} of $\mathcal{G}$.
\end{defn}

\begin{ex}\label{actiongpd}
A premier example of a groupoid is the \textbf{action groupoid} which can be associated to any smooth $G$-space $M$. Its set of objects is $\mathcal{G}_{0} = M$ and its set of arrows is $\mathcal{G}_{1} = M \times G$, so that we can write it as 
$$
M \times G \doublerightarrow{}{} M =: \mathcal{M}_{\mathcal{G}}
$$ 
In this case, $s(p,g) = p$ and $t(p,g) = g 
\cdot p$. Common notation for the action groupoid is $M//G$: the action groupoid is defined as a first step toward understanding coarse quotients which forget stabilizers or may not even be smooth, in the sense that the action of $G$ could fix certain points in $M$ and so $M/G$ could be singular.

\begin{defn}
A \textbf{Lie groupoid} $\mathcal{G} = \mathcal{G}_{1} \doublerightarrow{s}{t} \mathcal{G}_{0}$ is a groupoid such that both the space of arrows $\mathcal{G}_{1}$ and space of objects $\mathcal{G}_{0}$ are smooth manifolds, all structure maps are smooth, and the source and target maps $s,t : \mathcal{G}_{1} \to \mathcal{G}_{0}$ are surjective submersions. (In other words, a Lie groupoid is a groupoid internal to the category of smooth manifolds.)
\end{defn}

\begin{rmk}
Moreover, if $\pi: V \to M$ is a $G$-equivariant vector bundle, then we could also define \textit{its} action groupoid $V \times G \doublerightarrow{}{} V =: \mathcal{V}_{\mathcal{G}}$. Both $\mathcal{V}_{\mathcal{G}}$ and $\mathcal{M}_{\mathcal{G}}$ are in fact Lie groupoids, and $\mathcal{V}_{\mathcal{G}}$ is a vector space object over $\mathcal{M}_{\mathcal{G}}$ in the category of Lie groupoids. Thus, by some abuse of notation, we can view $\pi: \mathcal{V}_{\mathcal{G}} \to \mathcal{M}_{\mathcal{G}}$ as a vector bundle. 
\end{rmk}

\end{ex}

To consider the above definitions for infinite dimensional manifolds, we need to choose the right category: in our case, it is the category of Fr\'{e}chet manifolds. 

\begin{defn}
	A \textbf{Fr\'{e}chet Lie groupoid} is a groupoid internal to the category of Fr\'{e}chet manifolds: in other words, $\mathcal{G}_{1}$ and $\mathcal{G}_{0}$ are Fr\'{e}chet manifolds and $s$ and $t$ are smooth maps of Fr\'{e}chet manifolds. We denote their associated category as $\mathrm{FrLieGpd}$.
\end{defn}

\begin{ex}\label{metdiffgroupoid}
Combining Proposition \ref{fieldsasfrechetbundle} with the above definition implies that for a compact smooth manifold $X$,
\begin{equation}
	\mathrm{Met}(X) \times \mathrm{Diff}(X) \doublerightarrow{}{} \mathrm{Met}(X) =: \mathscr{M}//\mathscr{D}
\end{equation}
is a Fr\'{e}chet Lie groupoid. By Definition \ref{gencov} and the preceding remark, any generally covariant BV field theory constitutes a dg vector bundle or $L_{\infty}$ bundle of Fr\'{e}chet Lie groupoids:
\begin{equation}
	\pi : (\mathscr{F}//\mathscr{D}, \{S,-\}) \to \mathscr{M}//\mathscr{D}.
\end{equation}

\end{ex}

\begin{rmk}
Groupoids are a cornerstone in the definition of stacks, which are the spaces which we eventually would like to replace ordinary manifolds with for the purpose of including quotient data in the space. The soul of the matter lies within the definition of a prestack, which is motivated by the functor of points perspective. The difference is that instead of a functor from $\mathbf{Mfd}^{op}$ (or $\mathbf{CRing}$ for an algebraic geometer) to $\mathbf{Set}$, the functor lands in $\mathbf{Gpd}$, which contains any ``equivalence data" specific to the model at hand. We will define a stack and then quickly move onto the key example, to avoid unnecessary generalities.
\end{rmk}

\begin{defn}[Definition 1.1 in \cite{heinloth}]
A \textbf{prestack} is a 2-functor $\mathfrak{X} : \mathbf{Mfd}^{op} \to \mathbf{Gpd}$. \\
A prestack $\mathfrak{X}$ over $\mathbf{Mfd}^{op}$ is a \textbf{stack} if for any $N \in \mathbf{Mfd}^{op}$ and open cover $\{U_{i}\}$ of $N$, it \textbf{satisfies descent}, in other words: 
(1) Given objects $P_{i} \in \mathfrak{X}(U_{i})$ and isomorphisms $\varphi_{ij} : P_{i}\vert_{U_{i} \cap U_{j}} \to P_{j}\vert_{U_{i} \cap U_{j}}$ such that $\varphi_{jk} \circ \varphi_{ij} = \varphi_{ik}\vert_{U_{i} \cap U_{j} \cap U_{k}}$, there is an object $P \in \mathfrak{X}(N)$ and isomorphisms $\varphi_{i} : P\vert_{U_{i}} \to P_{i}$ such that $\varphi_{ij} = \varphi_{j} \circ \varphi_{i}^{-1}$. This is called \textbf{effective} descent data. \\
(2) Given $P, P' \in \mathfrak{X}(N)$ and isomorphisms $\varphi_{i} : P\vert_{U_{i}} \to P'\vert_{U_{i}}$ such that $\varphi_{i}\vert_{U_{i} \cap U_{j}} = \varphi_{j}\vert_{U_{i} \cap U_{j}}$, there is a unique map $\varphi : P \to P'$ such that $\varphi_{i} = \varphi\vert_{U_{i}}$.
\end{defn}

\begin{rmk}
As we have defined it, the above is actually a \textit{differentiable stack}; however, because this is the only kind of stack we need, we usually drop the adjective. 
\end{rmk}

\begin{ex}
A fundamental example of a stack over $\mathbf{Mfd}^{op}$ is an ordinary manifold. For such a manifold $M$, we can define the stack $\underline{M}$ as $\underline{M}(N) := \mathrm{Map}(N,M) = C^{\infty}(N, M)$ for $N \in \mathbf{Mfd}$. This stack is presented by the ``discrete groupoid" $M \times \{1\} \doublerightarrow{}{} M$, whose objects are the points of $M$ and the only morphisms are the identities for those points. This embeds $\mathbf{Mfd}$ into the category $\mathbf{Stk}$ of differentiable stacks. We also get the following essential lemma.
\end{ex}

\begin{thm}[Yoneda Lemma for Stacks: Lemma 1.3 in \cite{heinloth}]\label{stackyoneda}
	Let $\mathfrak{X}$ be a stack and let $M$ be a manifold. We have the following equivalence of categories:
	$$
	\mathfrak{X}(M) \cong \mathrm{Mor}_{\mathbf{Stk}}(\underline{M}, \mathfrak{X}). 
	$$
\end{thm}

Since stacks are designed to generalize the notion of an ordinary manifold, there should be an analogous notion of an atlas for stacks.\footnote{One could even define stacks by first constructing atlases, as is done for manifolds.} We may sometimes denote $\underline{M}$  simply as $M$, when it's implicit in context.

\begin{defn}
	An \textbf{atlas} (or covering) for a stack $\mathfrak{X}$ is a manifold $X$ and map $p : X \to \mathfrak{X}$ such that (1) for any manifold $Y$ and $Y \to \mathfrak{X}$, the stack $X \times_{\mathfrak{X}} Y$ is a manifold, and (2) $p$ is a submersion, i.e. for all $Y \to \mathfrak{X}$, the projection $Y \times_{\mathfrak{X}} X \to Y$ is a submersion.
\end{defn}

We now specify the most important example of a stack for this paper.

\begin{defn}\label{quotientstackdef}
Given a smooth $G$-manifold $M$, the associated \textbf{quotient stack} is the functor $[M/G] : \mathbf{Mfd}^{op} \to \mathbf{Gpd}$ such that the objects of $[M/G](N)$ are pairs $(P \xrightarrow{\pi} N, P \xrightarrow{\alpha} M)$, $\pi$ being a principal $G$-bundle over $N$ and $\alpha$ being a $G$-equivariant map, and the arrows are isomorphisms of principal $G$-bundles over $N$ commuting with the $G$-equivariant maps to $M$.
\end{defn}

Note that $[M/G]$ evaluated on a point recovers $M//G$, so that $[M/G]$ rightly gives a natural generalization of $M//G$. An atlas for $[M/G]$ is $M \to [M/G]$. Much like how we use atlases to define principal and vector bundles over an ordinary manifold, we use atlases to define such bundles over stacks, as follows.

\begin{defn}
	A principal $G$-bundle $\mathscr{P} \to \mathfrak{X}$ is given by a $G$-bundle $\mathscr{P}_{X}$ over an atlas $X \to \mathscr{P}$ with an isomorphism of the two pullbacks $p_{1}^{*}\mathscr{P}_{X} \xrightarrow{\simeq} p_{2}^{*}\mathscr{P}_{X}$ from $X \times_{\mathfrak{X}} X \to X$ satisfying the cocycle condition on $X \times_{\mathfrak{X}} X \times_{\mathfrak{X}} X$. 
\end{defn}

\begin{rmk}
	The definition of a vector bundle over a stack $\mathfrak{X}$ is completely analogous to this. Of course, one could instead invoke that a vector bundle $V \to \mathfrak{X}$ of rank $n$ is equivalent to a principal $GL(n, \mathbf{R})$-bundle and then use the preceding definition, anyway.
\end{rmk}

\begin{ex}\label{BG}
An essential example derived from the above definitions is that of $[\mathrm{pt}/ G]$, for $\mathrm{pt}$ a point. Applying the definition shows that $[\mathrm{pt}/ G](X)$ is precisely $\textrm{Bun}_{G}(X)$, the category of principal $G$-bundles over $X$ (any morphism of $G$-bundles over the same base space is necessarily an isomorphism). Because of this, it is common to identify $[\mathrm{pt}/ G]$ with $BG$, since $[X, BG]$ is equivalent to $\textrm{Bun}_{G}(X)$ modulo bundle ismomorphisms. 

Notice moreover that defining a vector bundle $V \to [\mathrm{pt}/G]$ amounts to fixing the vector space $V$ (a vector bundle over the point) \textit{as well as} a representation $\rho : G \to \mathrm{End}(V)$. In other words, we have an equivalence of categories: 
$$
\mathrm{VectBun}([\mathrm{pt}/G]) \cong \mathrm{Rep}(G).
$$
\end{ex}

\begin{rmk}
The preceding example is a simple but beautiful illustration of how specifying a vector bundle $V$ over a quotient stack $[M/G]$ is equivalent to specifying a $G$-equivariant vector bundle over a $G$-manifold $M$. We thus get the following fact.

\begin{thm}[Adapted from Example 4.5 in \cite{heinloth}]\label{vectbunequiv}
For $M$ a smooth $G$-space, we have the following equivalence of categories:
$$
\mathrm{VectBun}([M/G]) \cong \mathrm{VectBun}_{G}(M). 
$$	
\end{thm}

This is stated in \cite{heinloth} for cartesian sheaves on $[M/G]$ and we are representing a vector bundle by its space of sections to deduce the above, so the statement in \cite{heinloth} holds for a larger class of objects. Additionally, although it is outside the scope of this paper, it is worth mentioning that a quotient stack $[M/G]$ contains ``all possible ways" in which it could have been defined starting with an action groupoid: more than one groupoid could \textit{present} a stack (this is described and expounded on in \cite{carchedi}), so it is reassuring that the stack itself contains this data. 
\end{rmk}

A deeper level of care must be taken for our motivating example $[\textrm{Met}(X)/\textrm{Diff}(X)] = [\mathscr{M/D}]$, which is presented by the Fr\'{e}chet Lie groupoid $\mathscr{M//D}$. Smooth maps from an ordinary finite dimensional manifold to a Fr\'{e}chet manifold are well-defined, and so the associated maps from one to the other when viewed as their respective discrete groupoids are also well-defined. This allows us to formulate the following definition: 

\begin{defn}\label{metdiffmodstack}
 For a compact manifold $X$, let $[\textrm{Met}(X)/\textrm{Diff}(X)] = [\mathscr{M/D}] : \mathbf{Mfd}^{op} \to \mathbf{Gpd}$ be the functor such that the objects of $[\mathscr{M/D}](N)$ are pairs $(P \xrightarrow{\pi} N, P \xrightarrow{\alpha} \mathscr{M})$, $\pi$ being a principal $\mathscr{D}$-bundle over $N$ and $\alpha$ being a $\mathscr{D}$-equivariant map, and the arrows are isomorphisms of principal $\mathscr{D}$-bundles over $N$ commuting with the $\mathscr{D}$-equivariant maps to $\mathscr{M}$: $[\mathscr{M/D}]$ is the \textbf{moduli stack of metrics modulo diffeomorphism}.
 \end{defn}

 \begin{lem}\label{M//Dpresentsmodstack}
	For a compact manifold $X$, the Fr\'{e}chet Lie groupoid $\mathscr{M//D}$ presents the Fr\'{e}chet moduli stack $[\mathscr{M/D}]$. 
\end{lem}

\begin{rmk}
Definition \ref{metdiffmodstack} is stated in the ``ordinary" sense, so that we don't specify the Fr\'{e}chet nature of the manifolds. Then, Lemma \ref{M//Dpresentsmodstack} implies that $[\mathscr{M/D}]$ is represented via the canonical functor $\mathbf{FrGpd} \to \mathbf{Stk}$ sending a Fr\'{e}chet Lie groupoid to its associated differentiable stack. The very detailed paper \cite{robertsvozzo}--in particular Sections 2 and 5--provides additional details and examples for these definitions, and is what we primarily relied on above.
\end{rmk}

In light of Definition \ref{gencov} and Theorem \ref{vectbunequiv} as well as the preceding definition, we have:

\begin{prop}\label{GCforstacks}
	Any generally covariant family $\pi: (\mathscr{F}, \{S,-\}) \to \mathscr{M}$ of BV field theories descends to a Fr\'{e}chet dg vector bundle or $L_{\infty}$ bundle of stacks:
	\begin{equation}
		\pi: ([\mathscr{F/D}], \{S,-\}) \to [\mathscr{M/D}].
	\end{equation}
\end{prop}

\begin{rmk}
Conversely, any such bundle defines a generally covariant theory: in this sense, Proposition \ref{GCforstacks} can be taken as the definition of a generally covariant theory. Moreover, we mindfully dropped the notation of a fixed smooth manifold $X$ in the statement of this proposition: in the long run, we would like to better understand what kind of functor $([\mathscr{F/D}], \{S,-\}) \to [\mathscr{M/D}]$ constitutes from $\mathbf{Riem}_{n}$ to the category of $L_{\infty}$ bundles over stacks. More will be said on this in the following example. 
\end{rmk}

\begin{ex}[Perturbative Yang-Mills Theory]\label{perturbativeYM}
	The advantages of the stacky formulation of general covariance may be more convincing when considering theories which have more data involved; e.g. those with local symmetries. As an example, let us consider Yang-Mills theory: to begin, let $(X,g)$ be an oriented, $n$-dimensional Riemannian manifold, and let $G$ be a compact Lie group whose Lie algebra $\mathfrak{g}$ has a nondegenerate invariant pairing, $\langle -,- \rangle_{\mathfrak{g}}$. To minimize any topological complications, fix a \textit{trivial} principal $G$-bundle $P \to X$. 

In this instance, the fields for Yang-Mills theory are the connection one-forms $A \in \Omega^{1}(X, \mathfrak{g}) = \Omega^{1}(X) \otimes \mathfrak{g}$ associated to $P$, which constitute an affine Fr\'{e}chet space. To such a field, we can associate its curvature form $F_{A} := dA + \frac{1}{2}[A,A] \in \Omega^{2}(X,\mathfrak{g})$. Letting $\langle -,- \rangle$ denote the pairing on $\Omega^{\bullet}(X, \mathfrak{g})$ defined by
\begin{equation} \label{YMpairing}
\langle \omega_{1} \otimes E_{1}, \omega_{2} \otimes E_{2}\rangle := \int_{X}\omega_{1} \wedge \omega_{2} \langle E_{1},E_{2} \rangle_{\mathfrak{g}},	
\end{equation}
the Yang-Mills action functional can be written as 
\begin{equation}
	S_{YM}(A) = \frac{1}{2} \langle F_{A}, \star F_{A} \rangle, \footnote{In the case that the pairing on $\mathfrak{g}$ is the Killing form--as is the case for many physically relevant gauge theories--the action is usually denoted
	$$
	\frac{1}{2}\int_{X}\Tr (F_{A} \wedge \star F_{A}).
	$$}
\end{equation}
where $\star$ denotes the Hodge star operator. The Euler-Lagrange equations for this action are
\begin{equation}\label{YMequations}
	d_{A} \star F_{A} = 0,
\end{equation}
where $d_{A} = d + A$ is the exterior covariant derivative associated to $A$. Alternatively, this can be written as $(d_{A} \star d_{A}) A = 0$. 

To move toward the derived-geometric set up in the BV formalism, we must also consider that there is an action of the \textit{gauge group} $C^{\infty}(X,G)$ on the fields $\Omega^{1}(X,\mathfrak{g})$ defined such that for $g \in C^{\infty}(X,G)$, $g \cdot A$ is $A^{g} := g^{-1}Ag + g^{-1}dg$. $S_{YM}(A)$ is invariant under this action, and so the Yang-Mills equations are covariant with respect to it. Moreover, the infinitesimal gauge action is: for $\alpha \in C^{\infty}(X) \otimes \mathfrak{g} = \Omega^{0}(X,\mathfrak{g})$, $A \mapsto d_{A}\alpha = d\alpha + [A, \alpha] \in T_{A}\Omega^{1}(X, \mathfrak{g}) \cong \Omega^{1}(X, \mathfrak{g})$, the tangent space to the space of connection one-forms at $A$. This action suggests that we consider the tangent complex\footnote{This will be defined precisely later, in Definition \ref{tancpxprop}.} to $A$ as a point in the \textit{stack} of connections modulo gauge:
\begin{equation}
	\mathbf{T}_{A}[\Omega^{1}(X,\mathfrak{g})/\Omega^{0}(X,\mathfrak{g})] \cong \Omega^{0}(X,\mathfrak{g})[1] \xrightarrow{d_{A}} \Omega^{1}(X,\mathfrak{g}).
\end{equation}
We can begin to define a BV theory for perturbative Yang-Mills about a fixed solution $A$ by computing the $-1$-shifted cotangent bundle of the above:
\begin{equation}
	\Omega^{0}(X,\mathfrak{g})[1] \xrightarrow{d_{A}} \Omega^{1}(X,\mathfrak{g}) \xrightarrow{d_{A} \star_{g} d_{A}} \Omega^{n-1}(X, \mathfrak{g})[-1] \xrightarrow{d_{A}} \Omega^{n}(X, \mathfrak{g})[-2] =: \mathscr{E}_{(g,A)}.
\end{equation}
The shifted symplectic pairing comes from (\ref{YMpairing}) and the differential between $\Omega^{1}(X, \mathfrak{g})$ and $\Omega^{n-1}(X, \mathfrak{g})$ comes from the equations of motion (\ref{YMequations}). Also, we are being pedantic in that we are labeling the Hodge star with the metric used to define it.

\begin{rmk}
There is a dependence on the metric in the middle differential (the Yang-Mills term), but we could in principle compute whether or not the entire complex is diffeomorphism equivariant: this amounts to checking whether or not the infinitesimal gauge invariance--described by the differential $d_{A}$ between $\Omega^{0}(X,\mathfrak{g})[1]$ and $\Omega^{1}(X,\mathfrak{g})$ and also between $\Omega^{n-1}(X,\mathfrak{g})[-1]$ and $\Omega^{n}(X,\mathfrak{g})[-2]$--is also diffeomorphism equivariant. Put plainly, showing the diffeomorphism equivariance of $\mathscr{E}_{(g,A)}$ proves that perturbative Yang-Mills theory is generally covariant as a theory which \textit{also} depends on connections modulo (infinitesimal) gauge.
\end{rmk}

To show that $\mathscr{E}_{(g,A)}$ is $\mathscr{D}$-equivariant, we must show that the diagram 
\[\begin{tikzcd}\label{YMdiagram}
	\Omega^{0}(X,\mathfrak{g})[1] && \Omega^{1}(X, \mathfrak{g}) && \Omega^{n-1}(X, \mathfrak{g})[-1] && \Omega^{n}(X, \mathfrak{g})[-2] \\
	\\
\Omega^{0}(X,\mathfrak{g})[1] && \Omega^{1}(X, \mathfrak{g})  && \Omega^{n-1}(X, \mathfrak{g})[-1] && \Omega^{n}(X, \mathfrak{g})[-2]
	\arrow["{d_{A}}", from=1-1, to=1-3]
	\arrow["{d_{A}\star_{g}d_{A}}", from=1-3, to=1-5]
	\arrow["{d_{A}}", from=1-5, to=1-7]
	\arrow["{f^{*}}", from=1-1, to=3-1]
	\arrow["{f^{*}}", from=1-3, to=3-3]
	\arrow["{f^{*}}", from=1-5, to=3-5]
	\arrow["{f^{*}}", from=1-7, to=3-7]
	\arrow["{d_{f^{*}A}}", from=3-1, to=3-3]
	\arrow["{d_{f^{*}A}\star_{f^{*}g}d_{f^{*}A}}", from=3-3, to=3-5]
	\arrow["{d_{f^{*}A}}", from=3-5, to=3-7]
\end{tikzcd}\]
commutes, for any diffeomorphism $f \in \mathscr{D}$. Notice that in the lower complex, the Hodge star is defined by the metric $f^{*}g$ and the fixed connection form is $f^{*}A$.  To begin, let $\alpha \in \Omega^{0}(X,\mathfrak{g})[1]$. We get
\begin{equation}
f^{*}(d_{A}\alpha) = f^{*}(d\alpha + A \wedge \alpha) = d(f^{*}\alpha) + f^{*}A \wedge f^{*}\alpha,
\end{equation}
because the exterior derivative $d$ is manifestly covariant and pullbacks commute with wedge products, even if the forms have $\mathfrak{g}$ coefficients: this equation is then equal to $(d+f^{*}A)(f^{*}\alpha) = d_{f^{*}A}(f^{*}\alpha)$, which proves that the first square commutes. Moreover, this same computation shows that the last square commutes, too. 

Next, let $\omega \in \Omega^{1}(X, \mathfrak{g})$. Then:
\begin{align}
	(d_{A}\star_{g}d_{A})\omega &= (d_{A}\star_{g}(d\omega + A \wedge \omega)) \\
	&= (d+A)(\star_{g}d\omega + \star_{g}(A\wedge\omega)) \\
	&= (d\star_{g}d\omega + d\star_{g}(A\wedge\omega) + A \wedge \star_{g}d\omega + A \wedge \star_{g}(A\wedge\omega)).
\end{align}
Before we consider the diffeomorphism action, notice that the $L_{\infty}$ structure can be read off from the last expression. Pulling the above back along $f \in \mathscr{D}$ results in 
\begin{equation}
	f^{*}((d_{A}\star_{g}d_{A})\omega) = f^{*}(d\star_{g}d\omega) + f^{*}(d\star_{g}(A\wedge\omega)) + f^{*}(A \wedge \star_{g}d\omega) + f^{*}(A \wedge \star_{g}(A\wedge\omega)),
\end{equation}
which, when considering that the pullback commutes with the wedge product and the manifest covariance of the Hodge star, is equal to 
\begin{equation}
d\star_{f^{*}g}d(f^{*}\omega) + d\star_{f^{*}g}(f^{*}A\wedge f^{*}\omega) + f^{*}A \wedge \star_{f^{*}g}d(f^{*}\omega) + f^{*}A \wedge \star_{f^{*}g}(f^{*}A\wedge f^{*}\omega).
\end{equation}
From here, we see that this is equal to $(d+f^{*}A)\star_{f^{*}g}(d+f^{*}A)(f^{*}\omega)$, which is what we wanted. We have therefore shown the following:

\begin{thm}\label{YMequivariance}
	The bundle of $L_{\infty}$ algebras $\mathscr{E}(X) \to \mathrm{Met}(X) \times \Omega^{1}(X, \mathfrak{g})$ with fibers 
	$$
	\mathscr{E}_{(g,A)}(X) = \Omega^{0}(X,\mathfrak{g})[1] \xrightarrow{d_{A}} \Omega^{1}(X,\mathfrak{g}) \xrightarrow{d_{A} \star_{g} d_{A}} \Omega^{n-1}(X, \mathfrak{g})[-1] \xrightarrow{d_{A}} \Omega^{n}(X, \mathfrak{g})[-2],
	$$
	is  $\mathrm{Diff}(X)$-equivariant. 
	In other words, perturbative Yang-Mills theory is generally covariant, and the above bundle descends to a bundle of stacks:
	\begin{equation}\label{YMbundle}
	\mathscr{E}(X) \to [(\mathrm{Met}(X) \times \Omega^{1}(X,\mathfrak{g})) / \mathrm{Diff}(X)].
		\end{equation}
\end{thm}

\begin{rmk}
We would like to remind the reader that in the style of Subsection \ref{functors}, we can abandon the specific choice of $X$ in (\ref{YMbundle}) and the result is a functor from $\mathbf{Riem}_{n}$ to the category of bundles of $L_{\infty}$ algebras over stacks. In this case, what replaces the post-composed functor $\mathrm{Obs^{cl}} : \mathbf{Riem}_{n} \to \mathbf{dgAlg}$ of Proposition \ref{functorial}? To be more specific, what happens when we input a bundle of $L_{\infty}$ algebras over a stack to output something in $\mathbf{dgAlg}$ instead of a lone $L_{\infty}$ algebra? It is of great interest to elaborate more on these functors in future work.
\end{rmk}

\begin{rmk}
	Each of the fibers $\mathscr{E}_{(g,A)}$ describes the formal neighborhood of a solution $A$ to the Yang-Mills equations as an element of the \textit{stack} $[\Omega^{1}(X,\mathfrak{g})/C^{\infty}(X,G)]$ of connections modulo gauge transformation, and with background metric $g \in \mathscr{M}$: in this sense, we can view the preceding bundle as parameterizing formal stacks describing solutions to Yang-Mills modulo gauge living over the stack of metrics modulo diffeomorphism. More on formal neighborhoods in $[\Omega^{1}(X,\mathfrak{g})/C^{\infty}(X,G)]$ can be found in the paper \cite{ellgwill} on spontaneous symmetry breaking. 
	
	However, we should note that the preceding equivariance computations work out perfectly well if we don't treat them perturbatively: after all, the Yang-Mills term in (\ref{YMequations}) is diffeomorphism equivariant with respect to both connection one-forms and metrics. The caveat is that by using $L_{\infty}$ algebras, we are implictly invoking Theorem 2.0.2 in \cite{lurie}, in which the correspondence between $L_{\infty}$ algebras and formal moduli spaces is specified; however, if we formally substitute $\omega = A$ in the above computations, the equivariance property holds for what is evidently the nonperturbative case. I am interested to see how this can be remedied further to have a globalized version of Theorem \ref{YMequivariance}.
\end{rmk}

\begin{rmk}
	Theorem \ref{YMequivariance} states a version of general covariance in which additional physical fields are inextricably linked to $\mathrm{Met}(X)$ in the moduli stack. Indeed, any tensor field (even when taking values in some Lie algebra, for example) is defined with regard to its behavior under diffeomorphisms. So then why do we state general covariance in terms of \textit{metrics} modulo diffeomorphism? The answer is that in the development of the theory of general relativity, a key observation was Einstein's equivalence principle. 
	
	In general relativity, metrics represent the dynamical variables of the gravitational field, but any freely falling observer in a gravitational field can \textit{choose coordinates} so that they are in an inertial frame: in other words, any metric can be altered \textit{by some diffeomorphism} to be locally Euclidean or Lorentzian. In this sense, metrics and diffeomorphisms are intimately related when specifying gravitational dynamics, and so we use the associated stack as a baseline for quantifying general covariance. Further details are provided in Section 3 of \cite{norton}.
\end{rmk}
\end{ex}

\section{Formal Stacks}

In Section \ref{the main result}, we will make a connection to the version of the classical Noether's Theorem as described in Chapter 12 of \cite{cosgwill2}. However, we must first cross the bridge from the world of global stacks as we defined them in Section \ref{groupoids} to the world of formal moduli spaces, which are examples of formal stacks. A key step is to associate to a (differential graded) equivariant vector bundle a vector bundle over a formal moduli space:\footnote{Formal moduli spaces are alternatively named ``formal moduli problems".} in our case, this formal moduli space is a formal neighborhood in a quotient stack. In Section \ref{the main result}, this formal moduli space will be the formal neighborhood of a fixed metric in the moduli stack of metrics modulo diffeomorphism.

\subsection{Tangent Complexes}

The goal of the next portion is to understand what object we can associate to a point in a stack that plays the analogous role of a tangent space to an ordinary manifold. These ``tangent complexes" are necessary to compute function rings on formal neighborhoods in stacks, making them locally ringed spaces.

\begin{const}
Let $M$ be a smooth $G$-space and let $\textrm{Stab}(p) \subseteq G$ be the stabilizer subgroup of $p \in M$. The $G$-orbit of $p$ thus looks like a copy of $G/\textrm{Stab}(p)$ lying in $M$. If we consider the map $t_{p} : G \to M$ defined as $t_{p}(g) = g \cdot p$,\footnote{$t_{p}$ is in fact the target map for the Lie groupoid $M \times G \rightrightarrows M$ with $p \in M$ fixed in $M \times G$.} then its differential $dt_{p}$ can be used to define a 2-term cochain complex of vector spaces:
\begin{equation}
0 \to \mathfrak{g}[1] \xrightarrow{dt_{p}} T_{p}M \to 0 =: \mathbf{T}_{p}[M/G],
\label{tancpx}
\end{equation}
where $\mathfrak{g}$ is in cohomological degree $-1$ and $T_{p}M$ is in degree $0$. Alternative notation is $\mathbf{T}_{p}[M/G] = (\mathfrak{g}[1] \oplus T_{p}M, dt_{p})$. Note that $\textrm{Stab}(p)$ could be discrete here, although that is not seen in $\mathbf{T}_{p}[M/G]$. We can also compute
$$
\textrm{ker}(dt_{p}) = H^{-1}(\mathbf{T}_{p}[M/G]) =  \textrm{Lie(Stab}(p)).
$$
Thus, if $H^{-1}(\mathbf{T}_{p}[M/G]) = 0$, then the coarse quotient $M/G$ is an ordinary manifold at that point, since the action is free nearby it. $H^{0}(\mathbf{T}_{p}[M/G])$ is the quotient of $T_{p}M$ by $\textrm{im}(dt_{p})$: it is the usual tangent space of the coarse quotient at points $p \in M$ where the action is free. As it turns out, this is exactly the tangent object we are looking for, as the notation suggests: a precise statement and further details are wonderfully detailed in \cite{anel}.
\end{const}

\begin{prop}\label{tancpxprop}
The \textbf{tangent complex} to the quotient stack $[M/G]$ at a point $[p]$ is exactly $\mathbf{T}_{p}[M/G]$ as defined in equation (\ref{tancpx}).
\end{prop}

This inspired the saying that ``smooth stacks are geometric spaces whose tangent spaces are complexes concentrated in nonpositive cohomological degree". In the case of quotient stacks, we're lucky to have a concrete way of realizing their associated tangent complexes.

\begin{rmk}
If we take the union of all the complexes $\mathbf{T}_{p}[M/G]$ over all $p \in M$, we get a complex of vector bundles over $M$:
$$
0 \to \underline{\mathfrak{g}} \xrightarrow{dt} TM \to 0,
$$
where $\underline{\mathfrak{g}} = M \times \mathfrak{g}$, considering that the base space $M$ is implicit. $\underline{\mathfrak{g}}$ is called the \textit{Lie algebroid} associated to the action Lie groupoid $M // G$, and $dt$ is called the \textit{anchor map} of the Lie algebroid. This is a primordial example of a Lie algebroid.
\end{rmk}

\begin{ex}\label{metdgla}
Consider the natural action of the group of diffeomorphisms $\mathscr{D}$ of a manifold $X$ on the space of Riemannian metrics $\mathscr{M}$: $t_{g}(f) = f^{*}g$. According to \cite{krieglmichor}, the Lie algebra of $\mathscr{D}$ at the identity diffeomorphism is $\mathrm{Vect}(X)$, the set of vector fields: this will be the degree $-1$ part of our tangent complex.

We know that $T_{g}\mathscr{M} \cong \Gamma(X, \textrm{Sym}^{2}(T_{X}^{\vee}))$, so that we can compute 
\begin{equation}
\mathbf{T}_{g}[\mathscr{M}/\mathscr{D}] = (\Gamma(X, T_{X})[1] \oplus \Gamma(X, \textrm{Sym}^{2}(T_{X}^{\vee})), dt_{g}).
\end{equation}
Then, given $V \in \Gamma(T_{X})$, $dt_{g}(V) = L_{V}g$, where $L_{V}g$ is the Lie derivative of $g$ along $V$: one can see this by considering the one-parameter family of diffeomorphisms $f = \textrm{exp}(tV)$--i.e. letting $V$ be the infinitesimal generator of $f$--and computing the derivative at $t=0$ of the action of $f$ on $g$. Not all diffeomorphisms can be written this way: after all, $\mathscr{D}$ is not even a simply connected Lie group. Even worse, there are diffeomorphisms which are infinitesimally close to the identity diffeomorphism which cannot be written as $\textrm{exp}(tV)$ for some $V$ \cite{krieglmichor}; however, we need not worry about this in what is to come, as will be explained in Section \ref{the main result}.
\end{ex}

\begin{rmk}
	We will now show the relevance of the above for field theories by introducing notation and a lemma: we will perform the relevant computations for the example of the $\mathscr{D}$-equivariant differential graded vector bundle $(\mathscr{F}, Q) \to \mathscr{M}$ with fibers $\mathscr{F}_{g} = C^{\infty}(X) \xrightarrow{Q_{g}} \textrm{Dens}(X)[-1]$ and differential $Q_{g}\varphi = \Delta_{g}\varphi\mathrm{vol}_{g}$. 
	let us call the actions described in Lemma \ref{scalequiv} $\tau_{\mathscr{M}} : \mathscr{D} \to \mathrm{Diff}(\mathscr{M})$ and $\tau_{\mathscr{F}} : \mathscr{D} \to \mathrm{Diff}(\mathscr{F})$. There is also an ``action" on the differential, sending $Q_{g}$ to $Q_{f^{*}g}$ for $f \in \mathscr{D}$: it clearly comes from the action of $\mathscr{D}$ on $\mathscr{M}$, but also nicely intertwines with the input and output of the differential, as described in general covariance. 
	
	To get the infinitesimal version of these actions we use computations similar to those in Example \ref{metdgla}, keeping in mind that $\mathrm{Lie}(\mathscr{D}) \cong \mathrm{Vect}(X)$. The map $\tau_{\mathscr{M}_{g}} : \mathrm{Vect}(X) \to T_{g}\mathscr{M}$ is what we already considered earlier, namely $V \mapsto L_{V}g$, and the fibers are similar: 
	\begin{lem}\label{pert}
	The action of $\mathscr{D}$ on the underlying graded vector space of $\mathscr{F}_{g}$ defines an action of $\mathrm{Vect}(X)$ on $\mathscr{F}_{g}$. It has a degree 0 part $\tau_{\mathscr{F}_{g}} : \mathrm{Vect}(X) \to T_{\varphi}C^{\infty}(X) \cong C^{\infty}(X)$ and a degree 1 part $\tau_{\mathscr{F}_{g}} : \mathrm{Vect}(X) \to T_{\mu}\mathrm{Dens}(X) \cong \mathrm{Dens}(X)$; they are, respectively, $V \mapsto L_{V}\varphi$ and $V \mapsto L_{V}\mu$. 
	\end{lem}
	
Although we provided this example for clarity, such a lemma holds for any generally covariant BV field theory as we defined it in Definition \ref{gencov}, since tangent vectors can be defined for a Fr\'{e}chet manifold by means of it being locally modeled by Fr\'{e}chet spaces. The most interesting and physically relevant detail which must be addressed is what happens to the differential $\{S_{g},-\}$ under this infinitesimal action: this will be the content of Section \ref{the main result}.
	\end{rmk}

\subsection{Chevalley-Eilenberg Cochains as Rings of Functions}

We start with an action of a finite dimensional Lie group $G$ on a finite dimensional manifold $M$, and then specialize to the case of $M = \mathbf{R}^{n}$ to consider some concrete computations. In the example of diffeomorphisms of a manifold $X$ acting on the space of Riemannian metrics on $X$, $\textrm{Met}(X) = \mathscr{M}$ is a convex cone in $\Gamma(X, \textrm{Sym}^{2}(T^{\vee}X))$, so that we will be eventually specializing these constructions to vector spaces or ``nice" subsets thereof anyway.

\begin{const}\label{local}
Let $\widehat{M}_{p}$ denote the formal neighborhood of $p \in M$, defined so that its ring of functions $\mathscr{O}(\widehat{M}_{p})$ is the jets of $\mathscr{O}(M) := C^{\infty}(M)$ at $p$, and denote the inclusion map $\hat{p} : \widehat{M}_{p} \to M$: this is equivalent to the restriction map $\mathscr{O}(M) \to \mathscr{O}(\widehat{M}_{p})$. It is known that $\mathscr{O}(\widehat{M}_{p}) \cong \widehat{\textrm{Sym}}(T_{p}^{\vee}M)$, the Taylor series ring around $p \in M$, although this isomorphism is not canonical. We will use the latter, and call the Taylor series ring $\widehat{\mathscr{O}}_{p}$ when unambiguous.


The action of $G$ on $M$ is defined by a map $P : G \to \textrm{Diff}(M)$. Taking its total derivative gives us a map $\rho : \mathfrak{g} \to \textrm{Vect}(M)$ of Lie algebras, where we choose to view $\textrm{Vect}(M)$ as derivations of $\mathscr{O}(M)$. We then restrict the action of $\textrm{Vect}(M)$ on $C^{\infty}(M)$ to get an action of $\textrm{Vect}(\widehat{M}_{p})$ on $C^{\infty}(\widehat{M}_{p}) \cong \widehat{\mathscr{O}}_{p}$. The differential on $\mathbf{T}_{p}[M/G]$ encodes $\rho : \mathfrak{g} \to \textrm{Vect}(M)$ at the point $p$ and thus on the formal neighborhood $\widehat{M}_{p}$ of $p$ since $\rho$ is a map of Lie algebras: this gives us $\mathfrak{g} \to \textrm{Vect}(\widehat{M}_{p})$. Noting that $\textrm{Vect}(\widehat{M}_{p}) \cong \textrm{Der}(\widehat{\mathscr{O}}_{p})$ recovers the action of $\mathfrak{g}$ on $\widehat{\mathscr{O}}_{p}$ via derivations, this allows us to define $C^{\bullet}(\mathfrak{g}, \widehat{\mathscr{O}}_{p}) \cong C^{\bullet}(\mathfrak{g}) \otimes \widehat{\mathscr{O}}_{p}$ in the traditional way. 
\end{const}

\begin{lem}\label{CECs} 
	Chevalley-Eilenberg (CE) cochains of the differential graded Lie algebra defined by shifting $\mathbf{T}_{p}[M/G]$ up one degree, denoted $C^{\bullet}(\mathfrak{g} \xrightarrow{dt_{p}} T_{p}M[-1])$, and $C^{\bullet}(\mathfrak{g}, \widehat{\mathscr{O}}_{p})$ are isomorphic as differential graded commutative algebras.
	Moreover, $C^{\bullet}(\mathfrak{g}, \widehat{\mathscr{O}}_{p})$ is the ring of functions on the formal neighborhood of $[p] \in [M/G]$.
\end{lem}
\begin{proof}
It is a quick exercise to show that the underlying graded commutative algebras of $C^{\bullet}(\mathfrak{g} \xrightarrow{dt_{p}} T_{p}M[-1])$ and $C^{\bullet}(\mathfrak{g}, \widehat{\mathscr{O}}_{p})$ are identical, as long as one is careful to employ the noncanonical isomorphism $\mathscr{O}(\widehat{M}_{p}) \cong \widehat{\textrm{Sym}}(T_{p}^{\vee}M)$. From there, it is sufficient to show that the Chevalley-Eilenberg differentials are equivalent, which is also left as a brief exercise. 
\end{proof}

\begin{rmk}
This lemma implies that the dg Lie algebra 
\begin{equation}
	\mathfrak{g}_{p} := (\mathfrak{g} \oplus T_{p}M[-1], dt_{p},  [-,-]_{\mathfrak{g}})
\end{equation}
is of utmost importance. To say more about this, we must state a definition: 

\begin{defn}[Definition 3.1.2 in \cite{cosgwill2}]\label{formalmod}
A \textbf{formal (pointed) moduli problem} over $k$ is a functor of simplicially enriched categories
$$
F : \mathbf{dgArt}_{k} \to \mathbf{sSets},
$$	
where $\mathbf{dgArt}_{k}$ is the category of (local) Artinian dg algebras over $k$ and $\mathbf{sSets}$ the category of simplicial sets, which satisfies:
(1) $F(k)$ is contractible. (2) $F$ takes surjective maps in $\mathbf{dgArt}_{k}$ to fibrations in $\mathbf{sSets}$. 
(3) For $A,B,C \in \mathbf{dgArt}_{k}$ and surjections $B \to A$ and $C \to A$ (meaning we can define the fiber product $B \times_{A} C$), we require that the following natural map is a weak equivalence:
$$
F(B \times_{A} C) \to F(B) \times_{F(A)} F(C).
$$
\end{defn}
Clearly, this can be viewed as a ``localization" of the traditional algebro-geometric definition of a stack as a functor $\mathbf{CRing} \to \mathbf{Gpd}$ satisfying descent. What follows in the rest of this section and in Section \ref{the main result} depends on the following theorem, which allows us to connect the above objects to the more concrete dg Lie algebras and $L_{\infty}$ algebras we use for computations:

\begin{thm}[Theorem 2.0.2 in \cite{lurie}]\label{deformationtheory}
	There is an equivalence of $(\infty,1)$-categories between the category $\mathbf{Lie}_{k}$ of differential graded Lie algebras over a characteristic zero field $k$ and the category $\mathbf{Moduli}_{k}$ of formal pointed moduli problems over $k$. 
\end{thm}

The homotopy category of $L_{\infty}$ algebras is equivalent to the homotopy category of dg Lie algebras, so that the above remains true in that case (as is relevant for us). Theorem \ref{deformationtheory} confirms that the dg Lie algebra $\mathfrak{g}_{p}$ completely defines the data of the formal neighborhood of $[p]$ in $[M/G]$, as we suspected from Lemma \ref{CECs}.
\end{rmk}

\begin{rmk}
Much like how a quotient stack ``builds in" group action data into its definition, functions on a formal neighborhood $\widehat{[M/G]}_{p}$ in the stack, namely $C^{\bullet}(\mathfrak{g}, \hat{\mathscr{O}}_{p})$, have ``built into" them all of the $\mathfrak{g}$-invariant data. Concretely, $C^{\bullet}(\mathfrak{g}, \hat{\mathscr{O}}_{p})$ has the usual ring of functions $\hat{\mathscr{O}}_{p}$ as a subset: tensoring with $\textrm{Sym}(\mathfrak{g}^{\vee}[-1])$ and imposing the Chevalley-Eilenberg differential remembers the data of $\mathfrak{g}$ acting on $\widehat{M}_{p}$, and therefore on $\mathscr{O}(\widehat{M}_{p}) \cong \hat{\mathscr{O}}_{p}$ as well. To see how these ideas unfold in action, we refer the reader to Appendix \ref{findim}. 
\end{rmk}

\begin{ex}\label{metdgla2}
	In light of this lemma and Example \ref{metdgla}, the  (Fr\'{e}chet) dg Lie algebra we must consider in the context of general covariance is thus
\begin{equation}\label{metdgla3}
	\mathfrak{g}_{g} := \mathfrak{g}_{g}(X) :=  \Gamma(X, T_{X}) \xrightarrow{L_{\bullet}g} \Gamma(X, \mathrm{Sym}^{2}(T^{\vee}_{X}))[-1].
\end{equation}
Note that if don't specify evaluation of $\mathfrak{g}_{g}$ on all of $X$, it becomes a sheaf on the site $\mathbf{Riem}_{n}$ introduced  in Subsection \ref{functors}. The metric $g$ should also not be specified in that case, but we leave it in the notation to distinguish the above from a generic dg Lie algebra. 

By applying Lemma \ref{CECs}, we see that the ring of functions on the formal neighborhood of $[g] \in [\mathscr{M}/\mathscr{D}]$ is $C^{\bullet}(\mathrm{Vect}(X), \mathscr{O}(\widehat{\mathscr{M}}_{g})) \cong C^{\bullet}(\mathrm{Vect}(X)) \otimes  \widehat{\mathrm{Sym}}(T_{g}^{\vee}\mathscr{M})$, which we interpret as the derived invariants of $\mathscr{O}(\widehat{\mathscr{M}}_{g})$ with respect to the $\mathscr{D}$-action. Our definition of general covariance from earlier when properly ``localized" would imply that the observables of such a field theory $\mathscr{F}_{g}$ over $g \in \mathscr{M}$ form a module over $C^{\bullet}(\mathrm{Vect}(X), \mathscr{O}(\widehat{\mathscr{M}}_{g})) = C^{\bullet}(\mathfrak{g}_{g})$: this is exactly what is shown in Proposition \ref{themainthm}.
\end{ex}

\begin{rmk}
It should be noted that because $\mathrm{Vect}(X)$ and $T_{g}^{\vee}\mathscr{M}$  are infinite dimensional, the definition of $C^{\bullet}(\mathfrak{g}_{g}) \cong C^{\bullet}(\mathrm{Vect}(X)) \otimes  \widehat{\mathrm{Sym}}(T_{g}^{\vee}\mathscr{M})$ is not precisely the one from the finite dimensional case: in particular, we have shown that $\mathrm{Vect}(X)$ is a Fr\'{e}chet Lie algebra and $T_{g}^{\vee}\mathscr{M}$ is a Fr\'{e}chet vector space, and so the tensor product used for the preceding objects is the completed projective tensor product used in Definition \ref{functionals}. In this way, $C^{\bullet}(\mathfrak{g}_{g})$ represents the same data as it would if its inputs were finite dimensional, but we are just a bit more careful with functional analytic issues to ensure that all of the rings are well defined. 
\end{rmk}

\subsection{Vector Bundles over a Formal Stack}

Now that we have made things concrete with an example, we'd like to understand vector bundles in this context. We're primarily concerned with perturbative computations (those in the style of Construction \ref{local}); however, we will present the global picture first, since general covariance is first presented in such a context.

\begin{const}

Let $V$ be a $G$-equivariant vector bundle over $M$, for which the action $\tau_{M} : G \to \mathrm{Diff}(M)$ is not necessarily free. Call the action on the total space of $V \to M$ $\tau_{V} : G \to \mathrm{Diff}(V)$. Recall from Example \ref{actiongpd} that we get the pair of Lie groupoids $\mathcal{V}_{\mathcal{G}}$ and $\mathcal{M}_{\mathcal{G}}$ with a map $\pi : \mathcal{V}_{\mathcal{G}} \to \mathcal{M}_{\mathcal{G}}$ between them. This information in turn presents a pair of stacks, and the projection gives us a map $\pi: [V/G] \to [M/G]$ between those stacks. Here, $[V/G]$ is a vector space object in the category of stacks over the stack $[M/G]$, much like how $\mathcal{V}_{\mathcal{G}}$ is a vector space object in the category of Lie groupoids over the Lie groupoid $\mathcal{M}_{\mathcal{G}}$.

The action of a finite dimensional Lie group $G$ on a finite dimensional $M$ restricts to an action of the formal group $\widehat{G} \overset{\textrm{exp}}{\cong}
 \mathfrak{g}$ (defined as the formal neighborhood of the identity in $G$) on $\widehat{M}_{p}$, the formal neighborhood of $p \in M$. This defines a formal Lie groupoid which then presents the stack $[\widehat{M}_{p}/\widehat{G}] \cong \widehat{[M/G]}_{p}$, whose rings of functions we computed earlier to be $C^{\bullet}(\mathfrak{g}_{p})$, so that $\mathfrak{g}_{p}$ is the dg Lie algebra associated to the formal moduli problem $\widehat{[M/G]}_{p}$.
 

 We can pull back the $G$-equivariant vector bundle $V \to M$ along $\hat{p} : \widehat{M}_{p} \to M$ to get a $\mathfrak{g}$-equivariant vector bundle $\hat{p}^{*}V \to \widehat{M}_{p}$. Topologically, the total space of $\hat{p}^{*}V$ is the formal neighborhood of the entire fiber $\pi^{-1}(p) = V_{p}$, which we can think of heuristically as $V_{p} \times \widehat{M}_{p}$. Both parts of this product have an action of $\widehat{G}$, even though one of the directions is a formal space and the other a vector space which is not necessarily viewed as formal (i.e. its ring of functions may be polynomials, and not power series). Thus, we can consider the associated formal Lie groupoid here as well, and it presents the stack $[(\hat{p}^{*}V)/\widehat{G}]$. 
 
 The vector bundle which plays the local role of the global stack $[V/G] \to [M/G]$ is therefore 
$$
[(\hat{p}^{*}V)/\widehat{G}] \to [\widehat{M}_{p}/\widehat{G}].
$$
On account of $C^{\bullet}(\mathfrak{g}_{p})$ being the space of functions on $[\widehat{M}_{p}/\widehat{G}]$, we see that a section $\sigma : [\widehat{M}_{p}/\widehat{G}] \to [(\hat{p}^{*}V)/\widehat{G}]$ is an element of $C^{\bullet}(\mathfrak{g}_{p}) \otimes V_{p}$. This is the stackified and deformation-theoretic version of a section of $V \to M$ being an element of $C^{\infty}(M) \otimes V_{p}$ in local coordinates near $p$. Moreover, this reasoning  results in the following lemma.
\end{const}


\begin{lem}
	The ring of functions on $[(\hat{p}^{*}V)/\widehat{G}] \cong [(\hat{p}^{*}V)/\mathfrak{g}]$ is $C^{\bullet}(\mathfrak{g}, \widehat{\mathscr{O}}_{p} \otimes \mathrm{Sym}(V_{p}^{\vee})) \cong C^{\bullet}(\mathfrak{g}_{p}, \mathrm{Sym}(V_{p}^{\vee}))$, which is isomorphic as a graded ring to  $C^{\bullet}(\mathfrak{g}_{p}) \otimes \mathrm{Sym}(V_{p}^{\vee})$.
\end{lem}

\begin{proof}
The definition of $[(\hat{p}^{*}V)/\mathfrak{g}]$ implies that $\mathscr{O}([(\hat{p}^{*}V)/\mathfrak{g}])$ must be the derived $\mathfrak{g}$-invariant functions on the space $\widehat{M}_{p} \times V_{p}$. Given that $\mathscr{O}(\widehat{M}_{p} \times V_{p}) = \widehat{\mathscr{O}}_{p} \otimes \mathrm{Sym}(V_{p}^{\vee})$ and that both parts of this tensor product are $\mathfrak{g}$-modules, we can define the CE cochains $C^{\bullet}(\mathfrak{g}, \widehat{\mathscr{O}}_{p} \otimes \mathrm{Sym}(V_{p}^{\vee}))$. In conjunction with Lemma \ref{CECs}, these are the derived $\mathfrak{g}$-invariant functions we are looking for. To see that the differential graded rings are isomorphic, we simply note that the CE differential on both is
$$
d_{CE} = [-,-]^{\vee}_{\mathfrak{g}} + \tau_{M_{p}}^{\vee} + \tau_{V_{p}}^{\vee},
$$
where $\tau_{M_{p}}^{\vee}$ and $\tau_{V_{p}}^{\vee}$ are the ``duals" (as in Appendix \ref{findim}) to the induced actions $\tau_{M_{p}}$ and $\tau_{V_{p}}$ on $\widehat{\mathscr{O}}_{p}$ and $\mathrm{Sym}(V_{p}^{\vee})$, respectively. 
\end{proof}

\section{Field Theories as Bundles over Formal Stacks}\label{the main result}

\subsection{Equivariant Observables}

Next we shall consider how $Q_{g}$ and more generally $\{S_{g},-\}$ behave under \textit{arbitrary} perturbations $g + \varepsilon h$, for $h \in T_{g} \mathscr{M}$, and then invoke that $h = L_{V}g$ comes from a vector field $V$ to see what the effect is. But we already have the machinery to do this! The preceding sentence amounts to pulling back the dg vector or $L_{\infty}$ bundle $(\mathscr{F},\{S,-\}) \to \mathscr{M}$ to be over the formal neighborhood of $g \in \mathscr{M}$, and seeing what the ``full differential" $\{S_{g + \varepsilon h},-\}$ looks like over this formal neighborhood. 

\begin{rmk}\label{assurance}
Although in finite dimensions we have $\widehat{G} \overset{\textrm{exp}}{\cong}
 \mathfrak{g}$, we mentioned in Example \ref{metdgla} that it was no longer the case that there is a bijection between the formal neighborhood of the identity diffeomorphism in $\mathrm{Diff}(X)$ and its Lie algebra $\mathrm{Vect}(X)$ of vector fields: this could ostensibly be cause for alarm. However, by Lurie's Theorem \ref{deformationtheory} it remains true even in the infinite dimensional case that the dg Lie algebra $\mathfrak{g}_{g}$ introduced earlier is in correspondence with the formal neighborhood of $[g] \in [\mathscr{M/D}]$. This will be an assurance in what follows.
\end{rmk}

\begin{const}
A family $(\mathscr{F}, \{S,-\})$ of BV field theories defined as a dg vector or $L_{\infty}$ bundle over $\mathscr{M}$ pulls back to an appropriate bundle over the formal neighborhood of $g \in \mathscr{M}$, denoted $\widehat{\mathscr{M}}_{g}$, where $\mathscr{O}(\widehat{\mathscr{M}}_{g}) = \widehat{\mathscr{O}}_{g} \cong \widehat{\mathrm{Sym}}(T_{g}^{\vee}\mathscr{M})$. Heuristically, this pullback looks like $\widehat{\mathscr{M}}_{g} \times \mathscr{F}_{g}$. 

We get an analogous pullback of stacks when the theory is generally covariant. In this case, the $\mathscr{D}$-equivariant bundle $(\mathscr{F},\{S,-\}) \to \mathscr{M}$ is equivalent to a bundle of stacks $([\mathscr{F}/\mathscr{D}], \{S,-\}) \to [\mathscr{M}/\mathscr{D}]$. If we consider an equivalence class of metrics $[g] \in [\mathscr{M}/\mathscr{D}]$ and fix its formal neighborhood, we can pull back $([\mathscr{F}/\mathscr{D}],\{S,-\})$ over this formal neighborhood. We denote the total space of this pullback as $\widehat{[\mathscr{F}/\mathscr{D}]}_{g}$. We can then conclude:
\begin{prop}\label{themainthm}
For a generally covariant family $\pi: ([\mathscr{F}/\mathscr{D}], \{S,-\}) \to [\mathscr{M}/\mathscr{D}]$ of BV classical field theories and for a fixed $[g] \in [\mathscr{M}/\mathscr{D}]$, we have
\begin{equation}
\mathscr{O}(\widehat{[\mathscr{F}/\mathscr{D}]}_{g}) \cong C^{\bullet}(\mathfrak{g}_{g}, \mathrm{Obs}^{\mathrm{cl}}(X, \mathscr{F}_{g})). 
\end{equation}
\end{prop}
\begin{proof}
By the equivalence of categories in Theorem \ref{deformationtheory} which we are taking for granted, the formal moduli space defined by the formal neighborhood of $[g] \in [\mathscr{M/D}]$ is equivalent to the dg Lie algebra 
	$$
	\mathfrak{g}_{g} = \Gamma(X, T_{X}) \xrightarrow{L_{\bullet}g} \Gamma(X, \mathrm{Sym}^{2}(T^{\vee}_{X}))[-1].
	$$
The dg algebra of functions on this formal neighborhood is thus
	$$
	C^{\bullet} (\mathfrak{g}_{g}) \cong C^{\bullet}(\mathrm{Vect}(X), \widehat{\mathrm{Sym}}(T_{g}^{\vee}\mathscr{M})).
	$$
The ring of functions on the fiber part of the pullback is simply $C^{\bullet}(\mathrm{Vect}(X), \mathrm{Obs}^{\mathrm{cl}}(X, \mathscr{F}_{g}))$, since it is $\mathscr{O}(\mathscr{F}_{g})$ with the differential $\{S_{g},-\}$ and the implicit action of $\mathrm{Vect}(X)$ on the theory and thus on its observables. Hence, the underlying dg ring of functions on $\widehat{[\mathscr{F}/\mathscr{D}]}_{g}$ is the underlying dg ring of $C^{\bullet}(\mathfrak{g}_{g}, \mathrm{Obs}^{\mathrm{cl}}(X, \mathscr{F}_{g}))$. Both dg rings have Chevalley-Eilenberg differential 
	\begin{equation}
d_{CE} = [-,-]^{\vee}_{\mathrm{Vect}(X)} + \tau_{\mathscr{M}_{g}}^{\vee} + \tau_{\mathscr{F}_{g}}^{\vee} + \{S_{g}, - \}.
\end{equation}
Here, the first three terms are the usual Chevalley-Eilenberg differential concerned with the dual of the bracket on $\mathrm{Vect}(X)$ and the actions of $\mathrm{Vect}(X)$ on $\widehat{\mathrm{Sym}}(T_{g}^{\vee}\mathscr{M})$ and $\mathscr{O}(\mathscr{F}_{g})$, and the fourth term is the differential on the free field theory over $[\widehat{\mathscr{M}}_{g}/\mathrm{Vect}(X)]$. Since the underlying rings agree and the CE differentials do, too, this gives the result.
\end{proof}

\begin{rmk}
	As in the case of ordinary manifolds, the ring of functions on the bundle is a module over the ring of functions on the base space. In fact, the veracity of the above claim can almost be taken as a definition: in the case where we treat the BV field theory $\mathscr{F}$ perturbatively, Proposition \ref{themainthm} simply computes the function ring of the formal moduli problem representing perturbative fields parameterized by a formal neighborhood in $[\mathscr{M/D}]$. As it stands, the statement also includes polynomial functions of nonperturbative free fields along the fibers.\footnote{The fibers thus constitute a ``non-formal" part of the total moduli problem.}

	Proposition \ref{themainthm} implicitly supplies a natural action of $\mathfrak{g}_{g}$ on $\mathrm{Obs}^{\mathrm{cl}}(X, \mathscr{F}_{g})$: this allows us to conclude that Noether's Theorem as presented in Theorem 12.4.1 of \cite{cosgwill2} applies.  We have not given the precise details of a ``full" $L_{\infty}$ action of $\mathfrak{g}_{g}$ on $\mathrm{Obs}^{\mathrm{cl}}(X, \mathscr{F}_{g})$, but its existence is implicit in the theorem: $\{S_{g},-\}$ as part of $d_{CE}$ above contains information about the formal neighborhood of $[g] \in [\mathscr{M/D}]$. We will provide a thorough description of how this goes momentarily. 
\end{rmk}
\end{const}

\begin{rmk}[Further remarks on functoriality]\label{morefunctors}
Before getting into explicit computations, we would like to mention in the vein of Subsection \ref{functors} that because $\mathfrak{g}_{g}$ is a sheaf on $\mathbf{Riem}_{n}$ (its diffeomorphism equivariance can quickly be checked), the equivariant observables similarly define a factorization algebra, as in Proposition \ref{functorial}:
$$
C^{\bullet}(\mathfrak{g}_{g}(-), \mathrm{Obs^{cl}}(-,\mathscr{F})) : \mathbf{Riem}_{n} \to \mathbf{dgVect}.
$$	
This provides yet another factorization algebra when evaluated on the site of Riemannian manifolds. Thus, considering the stacky geometry of $[\mathscr{M/D}]$ for a fixed $X$ and invoking $\mathbf{Riem}_{n}$-naturality after the fact once again provides an interesting construction (and generalization) of objects introduced, for example, in \cite{cosgwill2}, while simultaneously opening avenues of comparison to prevailing literature. 
\end{rmk}

\noindent \textbf{Note:} From here on out, we will be treating both free\footnote{Any dg Lie algebra is an $L_{\infty}$ algebra where the only nontrivial bracket is $\ell_{1}$.} and interacting theories as $L_{\infty}$ algebras, since the reliance on $L_{\infty}$ structures for defining actions becomes more important. In practice, this means we will be using $\mathscr{L}$ (alias $\mathscr{F}[-1]$) to denote the fields. 

\begin{rmk}
	Recall that when $\mathfrak{g}$ is a Lie algebra and $R$ is a $\mathfrak{g}$-module, $H^{0}(C^{\bullet}(\mathfrak{g}, R)) = R^{\mathfrak{g}}$, the $\mathfrak{g}$-invariants of $R$. Analogously, albeit with slightly more care to compute, we have:
	\begin{equation}
		H^{0}(C^{\bullet}(\mathfrak{g}_{g}, \mathrm{Obs}^{\mathrm{cl}}(X, \mathscr{L}_{g}))) = \{ F \in \mathrm{Obs}^{\mathrm{cl}}(X, \mathscr{L}_{g}) : F(g+\varepsilon L_{V}g) - F(g) = 0 \}.
	\end{equation}
	A prime example of such an $F$ is the action functional $S_{g}$ of any generally covariant theory. Moreover, if $V$ is a Killing field for $g$, the equality in the conditions on the right side above holds trivially: this is a shadow of the fact that a moduli stack ``remembers" stabilizers where the coarse quotient would not.
	
	Although it is meaningful (and a good sanity check) to compute cohomology groups, we stick to Noether's philosophy of focusing on the cochain complexes themselves. In our case, this means understanding what the equivariant observables are providing. There is a guiding definition which, when unpacked carefully, tells us the value of what we found above:
\end{rmk}

\begin{defn}[Definition 12.2.12 in \cite{cosgwill2}]\label{Linfinityaction}
	For $\mathfrak{g}$ a dg Lie or $L_{\infty}$ algebra and $\mathscr{L}$ an (elliptic) $L_{\infty}$ algebra encoding a Batalin-Vilkovisky classical field theory, an \textbf{action of $\mathfrak{g}$ on $\mathscr{L}$} is any of the following data: 
	(i) An $L_{\infty}$ structure on $\mathfrak{g} \oplus \mathscr{L}$ such that the exact sequence
	\begin{equation}
		\mathscr{L} \to \mathfrak{g} \ltimes \mathscr{L} \to \mathfrak{g}
	\end{equation}
	is a sequence of maps of $L_{\infty}$ algebras, the structure maps $\mathfrak{g}^{\otimes n} \otimes \mathscr{L}^{\otimes m} \to \mathscr{L}$ are polydifferential operators on the $\mathscr{L}$-variables, and the action preserves the pairing $\omega$. \\
	(ii) An $L_{\infty}$ morphism $\mathfrak{g} \to C^{\bullet}_{\mathrm{loc}}(\mathscr{L})[-1]$. \\
	(iii) A degree 1 element $S^{\mathfrak{g}}$ in the dg Lie algebra 
	$$
	\mathrm{Act}(\mathfrak{g},\mathscr{L}) := C^{\bullet}_{\mathrm{red}}(\mathfrak{g}) \otimes C^{\bullet}_{\mathrm{loc}}(\mathscr{L})[-1]
	$$
	which satisfies the Maurer-Cartan equation 
	\begin{equation}
		d_{\mathfrak{g}}S^{\mathfrak{g}} + d_{\mathscr{L}}S^{\mathfrak{g}} + \frac{1}{2} \{S^{\mathfrak{g}},S^{\mathfrak{g}}\} = 0:
	\end{equation}
	this can be viewed as an \textit{equivariant} classical master equation. 
\end{defn}

\begin{rmk}
	By $C^{\bullet}_{\mathrm{loc}}(\mathscr{L})$ above we mean observables for $\mathscr{L}$ that are local in the sense of Definition \ref{locals}: $C^{\bullet}_{\mathrm{loc}}(\mathscr{L})[-1]$ is the formal moduli version of symplectic vector fields, which control symmetries and deformations of a classical field theory. $C^{\bullet}_{\mathrm{red}}(\mathfrak{g})$ is defined as the kernel of the augmentation map $C^{\bullet}(\mathfrak{g}) \to \mathbf{R}$.\footnote{Thorough details concerning these two rings is provided in Chapters 3 and 4 of \cite{cosgwill2}.}
	
	Moreover, since $S^{\mathfrak{g}}$ is local in the fields $\mathscr{L}$, it defines a derivation of $\mathrm{Obs}^{\mathrm{cl}}(X, \mathscr{L})$ via $\{S^{\mathfrak{g}},- \}$: this is precisely what is used to define the action of $\mathfrak{g}_{g}$ on $\mathrm{Obs}^{\mathrm{cl}}(X, \mathscr{L}_{g})$ when computing the equivariant classical observables in Proposition \ref{themainthm}.
\end{rmk}

\begin{rmk}\label{equivariantfunctional}
	The facet of the preceding definition we will hone in on is the third one: finding a functional $S^{\mathfrak{g}}$ which satisfies an equivariant classical master equation provides a concrete computational representation of the action of $\mathfrak{g}$ on $\mathscr{L}$ and a more complete picture of the Chevalley-Eilenberg description of how the formal moduli stack acts on the theory.
	
	We would like to encode both deformations by $h \in T_{g}\mathscr{M} = \Gamma(X, \mathrm{Sym}^{2}(T_{X}^{\vee}))$ and an action of vector fields $V \in \mathrm{Vect}(X)$ on the BV field theory, since these are the degree 1 and 0 parts (respectively) of the dg Lie algebra $\mathfrak{g}_{g}$ representing the formal neighborhood of $g$ as an element of the stack $[\mathscr{M/D}]$. Any action functional $S_{g}$ for a BV theory can be written as 
	$$
	S_{g}(\phi) = \int_{X}\phi D_{g}(\phi),
	$$
	where the differential operator may be a nonlinear function in $\phi$. Denoting the antifields for the theory as $\psi$, we thus posit the following:
	\begin{equation}\label{stressenergytensor}
		S^{\mathfrak{g}_{g}} = \int_{X} \phi \big(D_{g+\varepsilon h} - D_{g} \big)(\phi) + \int_{X}(L_{V}\phi) \psi.
	\end{equation}
	On the right side, we interpret $D_{g+\varepsilon h}$ as a formal power series in the metric perturbation $h$, and by $L_{V}$ we mean the ``natural action" of vector fields on the fields, which are usually tensorial in nature (hence the notation). In accordance with (ii) in Definition \ref{Linfinityaction}, the $L_{\infty}$ morphism $\mathfrak{g} \to C^{\bullet}_{\mathrm{loc}}(\mathscr{L})[-1]$ is thus given by sending $(V,h) \in \mathfrak{g}_{g}$ to:
	$$
	\{S^{\mathfrak{g}_{g}},-\} \in C^{\bullet}_{\mathrm{loc}}(\mathscr{L})[-1] \subset\ \mathrm{Obs}^{\mathrm{cl}}(X, \mathscr{L}_{g})[-1],
	$$
	where the latter is interpreted as symplectic vector fields on $B\mathscr{L}$, the formal derived critical locus as seen in Remark \ref{derivedcriticallocus}. By means of general covariance, which implies that either the dg or $L_{\infty}$ structure prescribed by $D_{g}(\phi)$ is diffeomorphism equivariant, $S^{\mathfrak{g}_{g}}$ satisfies the necessary classical master equation. Strictly speaking, what we need in the preceding is equivariance with respect to the action by vector fields: in the case of the free scalar field, this comes from the fact that the Laplace-Beltrami operator satisfies (modulo $\varepsilon^{2}$)
	$$
	\Delta_{g+ \varepsilon L_{V}g} = \Delta_{g} + \varepsilon[L_{V},\Delta_{g}].\footnote{Details for this are provided in Equation (\ref{infgencov}) and the surrounding commentary below.}
	$$
\end{rmk}

\begin{rmk}
	What we have presented so far allows us to provide a more precise and meaningful expression for the Chevalley-Eilenberg differential for $C^{\bullet}(\mathfrak{g}_{g}, \mathrm{Obs}^{\mathrm{cl}}(X, \mathscr{L}_{g}))$ in Proposition \ref{themainthm}:
	\begin{equation}
	d_{CE}	= [-,-]^{\vee}_{\mathrm{Vect}(X)} + \{S^{\mathfrak{g}_{g}},-\} + \{S_{g},-\}.
	\end{equation}
This provides the usual dual to the action of vector fields on themselves; however, we have here the $L_{\infty}$ action of $\mathfrak{g}_{g}$ on the observables as well as the usual differential $\{S_{g},-\}$ of the observables on their own. If we consider, for example, the interacting scalar field with polynomial potential as in Example \ref{pertcon1}, the latter two brackets above would combine to result in bracketing with:
\begin{equation}
	S^{\mathrm{tot}} := S_{g} + S^{\mathfrak{g}_{g}} = \int_{X} \varphi Q_{g + \varepsilon h}\varphi + \sum_{n \geq 2} \int_{X} \frac{\lambda_{n}}{n!} \varphi^{n} \mathrm{vol}_{g+ \varepsilon h} + \int_{X} (L_{V}\varphi) \psi.
\end{equation}
In Definitions 1 and 2 of \cite{getzler}, Getzler defines covariance by supplying something like $S^{\mathrm{tot}}$ above and demanding that it satisfy a Maurer-Cartan equation for a \textit{curved} Lie (super) algebra: I would be curious to see how the connection between the two could be made completely precise.
\end{rmk}

To expound more on all of the above, we must introduce the \textit{stress-energy tensor}.

\subsection{The Stress-Energy Tensor}\label{stressenergymomentum} We'd like to consider the stress-energy tensor for the free scalar BV theory: its generalization to the polynomial self-interaction in Lemma \ref{polynomialinteraction} is computationally simple. This section is intended to see how an example of Definition \ref{Linfinityaction} plays out as well as connect between the above work to how things are ``usually done" in physics.

To begin, let us consider an arbitrary first order deformation of the Laplacian $\Delta_{g}$ on a Riemannian manifold $(X,g)$: in other words, let $g_{t}$ be a one-parameter family of metrics such that $g_{0} = g$ and let us compute
$$
\frac{d}{d t}\Delta_{g_{t}} \varphi \Big|_{t=0}. 
$$
Writing $\Delta_{g_{t}}$ in coordinates and not evaluating at $t=0$ for now, we have:
\begin{equation}\label{difflap}
\frac{d}{dt}\Big(\frac{1}{\sqrt{\mathrm{det}g_{t}}} \partial_{\mu}(\sqrt{\mathrm{det}g_{t}} g_{t}^{\mu\nu}\partial_{\nu}\varphi) \Big).
\end{equation}
Recall that for a one-parameter family of invertible matrices $A(t)$, we have
$$
\frac{d}{dt}\mathrm{det}A(t) = \Tr (A(t)^{-1}A'(t))\mathrm{det}A(t).
$$
Using this and a few other manipulations, expression (\ref{difflap}) reduces to 
\begin{equation}\label{deriv}
	\frac{-1}{2}\Tr(g_{t}^{-1}\partial_{t}g_{t})\Delta_{g_{t}}\varphi + \frac{1}{\sqrt{\mathrm{det}g_{t}}}\partial_{\mu}\Big( \frac{\sqrt{\mathrm{det}g_{t}}}{2}\Tr(g_{t}^{-1}\partial_{t}g_{t})g_{t}^{\mu\nu}\partial_{\nu}\varphi +  \sqrt{\mathrm{det}g_{t}}\partial_{t}g_{t}^{\mu\nu}\partial_{\nu}\varphi \Big).
\end{equation}
Denote the derivative of $g_{t}$ at $g_{0} = g$ as  $\delta g := \partial_{t}g_{t}|_{t=0}$ (this is the traditional notation in physics, although we could call it $h$ as a degree 1 element of $\mathfrak{g}$). Evaluating at $t=0$ gives:

\begin{lem}\label{derlaplace}
	The first order deformation of the Laplacian $\Delta_{g}$ with respect to the metric $g$ is
	\begin{equation}\label{Lapdef}
\frac{d}{d t}\Delta_{g_{t}} \varphi \Big|_{t=0} = \frac{-1}{2}\Tr(g^{-1}\delta g)\Delta_{g}\varphi + \frac{1}{\sqrt{\mathrm{det}g}}\partial_{\mu} \Big(\sqrt{\mathrm{det}g}\big(\frac{1}{2} \Tr(g^{-1}\delta g) g^{\mu\nu}\partial_{\nu} \varphi + \delta g^{\mu\nu}\partial_{\nu}\varphi\big) \Big).
	\end{equation}
Moreover, if we assume the deformation $\delta g \in T_{g}\mathscr{M}$ is induced by an isometry of $g$, then the first order deformation of the Laplacian is identically zero. 
\end{lem}

\begin{rmk}
Often in the physics literature, we write
$$
\Tr(g^{-1}\delta g) = g^{\mu\nu}\delta g_{\mu\nu},
$$
and we may occasionally adopt that notation. Moreover, the above computation is done with the action functional (\ref{scalfunctional}) in mind. Thus, any difference with the stress-energy tensor computations using the functional
\begin{equation}\label{physfunc}
\int_{X} g^{\mu\nu}\partial_{\mu}\varphi\partial_{\nu}\varphi \mathrm{vol}_{g} = \int_{X} g^{\mu\nu}\partial_{\mu}\varphi\partial_{\nu}\varphi\sqrt{\mathrm{det}g} d^{n}x,
\end{equation}
which is just as common in the physics literature, differs only by boundary terms.
\end{rmk}

First, we shall provide a general definition of the stress-energy tensor for any theory.

\begin{const}
	Let $S_{g} \in C^{\bullet}(\mathscr{L}_{g})$ be an action functional for the dg space of BV fields $\mathscr{L}_{g}$ which depends on a fixed background metric $g \in \mathscr{M}$. It can thus be written as
	$$
	S_{g}(\phi) = \int_{X} L_{g}(\phi),
	$$
	where $\phi \in \mathscr{L}$ and $L_{g}(\phi)$ is a Lagrangian density.  If we let $g_{t}$ be a one-parameter family of metrics such that $g_{0} = g$, we can perform computations similar to those in Lemma \ref{derlaplace} to compute:
	\begin{equation}\label{funcder}
	\frac{\delta}{\delta g} S_{g}(\phi) := \frac{d}{dt}S_{g_{t}}(\phi) \Big|_{t=0} = \int_{X} \frac{d}{dt}L_{g_{t}}(\phi) \Big|_{t=0}.
	\end{equation}
	The notation invoked on the left hand side is common in physics literature, and defined this way in \cite{wald}. Up to boundary terms which we can safely ignore, (\ref{funcder}) can be written as 
	\begin{equation}\label{dgT}
	\int_{X} \delta g^{\mu\nu} T_{\mu\nu}(g, \phi) \mathrm{vol}_{g},
	\end{equation}
	for some $T_{\mu\nu}(g, \phi)$ (or simply $T_{\mu\nu})$ which depends on the fields $\phi$ and on the metrics $g$.

	\begin{defn}\label{T}
	$T_{\mu\nu}$ is the \textbf{stress-energy} (or \textbf{energy-momentum}) \textbf{tensor}	of a field theory on $X$ with fields $\phi \in \mathscr{F}$ and action functional $S_{g}$ depending on $g \in \mathrm{Met}(X)$. 
	\end{defn}
\end{const}

\begin{rmk}
Before moving on, we must take note of an important fact: the stress-energy tensor coupled to the metric perturbation as in Equation (\ref{dgT}) is precisely the first order in $\varepsilon$ part of the power series in $h$ found in Equation (\ref{stressenergytensor}).
\end{rmk}

\begin{ex}
	To compute the stress-energy tensor of free massless scalar field, we will begin by noting that according to the definition, we must compute
	$$
	\int_{X} \varphi \frac{d}{dt}Q_{g_{t}}\varphi \Big|_{t=0},
	$$
	where $Q_{g}\varphi = \Delta_{g}\varphi \mathrm{vol}_{g}$. Lemma \ref{derlaplace} is useful, since we have already done the necessary work on the first piece. However, note that $Q_{g}\varphi$ is written in coordinates as 
	$$
	\partial_{\mu}(\sqrt{\mathrm{det}g} g^{\mu\nu}\partial_{\nu}\varphi)d^{n}x,
	$$
	so that we have in fact stripped away some of the complexity of the computation by pairing with the Riemannian volume form. Hence, we can use Lemma \ref{derlaplace} and toss away the first term to get
	\begin{equation}\label{SET}
		\int_{X} \varphi \frac{d}{dt}Q_{g_{t}}\varphi \Big|_{t=0} = \int_{X} \varphi \partial_{\mu} \Big(\sqrt{\mathrm{det}g}\big(\frac{1}{2} \Tr(g^{-1}\delta g) g^{\mu\nu}\partial_{\nu} \varphi + \delta g^{\mu\nu}\partial_{\nu}\varphi\big) \Big)d^{n}x.
	\end{equation}
	This is not yet in the preferred form in (\ref{dgT}), but if we integrate by parts and invoke that 
	$
	\Tr(g^{-1}\delta g) = g^{\mu\nu}\delta g_{\mu\nu} = g_{\mu\nu}\delta g^{\mu\nu},
	$
the above becomes 
	\begin{equation}\label{SET2}
		\int_{X} \delta g^{\mu\nu} \big(-\partial_{\mu}\varphi \partial_{\nu} \varphi - \frac{1}{2}g_{\mu\nu} (g^{\rho\sigma} \partial_{\rho}\varphi \partial_{\sigma}\varphi ) \big) \mathrm{vol}_{g},
	\end{equation}
	where we changed the labelling of indices in the second term to omit confusion. Thus, the stress-energy tensor for our example is $T_{\mu\nu} = -\partial_{\mu}\varphi \partial_{\nu} \varphi - \frac{1}{2}g_{\mu\nu} (g^{\rho\sigma} \partial_{\rho}\varphi \partial_{\sigma}\varphi )$. We would have computed this without any by-parts maneuvers had we started with the action functional (\ref{physfunc}) more common in physics literature, but it is a good exercise to see how these agree.
\end{ex}

The above is the traditional trajectory one takes to finding the stress-energy tensor; however, since our theory is generally covariant and so we can use facts about equivariant vector bundles to simplify things, let us consider what that buys us. To begin, let $f_{t}$ be a one-parameter subgroup of diffeomorphisms. General covariance implies that
\begin{equation}\label{infgencov}
\frac{d}{dt}\int_{X} (f_{t}^{*}\varphi)\Delta_{	f_{t}^{*}g} (f_{t}^{*}\varphi)\mathrm{vol}_{f_{t}^{*}g} \Big|_{t=0} = 0. 
\end{equation}
Unfolding the left hand side, this equation becomes
$$
\int_{X} (L_{V}\varphi)\Delta_{g}\varphi\mathrm{vol}_{g} + \int_{X} \varphi \Delta_{g}(L_{V}\varphi)\mathrm{vol}_{g} + \int_{X} \varphi (\frac{d}{dt}\Delta_{f^{*}_{t}g}\big|_{t=0})\varphi\mathrm{vol}_{g} + \int_{X} \varphi \Delta_{g}\varphi(\frac{d}{dt}\mathrm{vol}_{f^{*}_{t}g}\big|_{t=0}) = 0.
$$
Here, we assumed that $V$ generates the flow $f_{t}$, and used the fact that $\frac{d}{dt}(f^{*}_{t}\varphi)|_{t=0} = L_{V}\varphi$. This equation is an integrated linear approximation to the equivariance property computed in Lemma \ref{scalequiv}: it states concretely that a simultaneous first order perturbation along the $\mathscr{D}$-orbit in $\mathscr{M}$ and in $\mathscr{L}_{g}$ is trivial. 

\begin{rmk}
The third and fourth terms on the left hand side are exactly those that comprise the integral of the stress-energy tensor in the special case that the derivative is computed in the direction of the $\mathscr{D}$-orbit. This grants us two key insights: 

(1) Computationally, the above amounts to the metric perturbation (an element of $T_{g}\mathscr{M}$) coming from an infinitesimal diffeomorphism (i.e. a vector field). But we have seen this before: this is saying that $\delta g  \in T_{g}\mathscr{M}$ is in the image of the differential in the dg Lie algebra $\mathfrak{g}_{g}$ in Example \ref{metdgla2}. Hence, $\delta g^{\mu\nu} = L_{V}g^{\mu\nu}$ (the computation works fine even though $g^{\mu\nu}$ is technically the inverse). With this, Equation (\ref{SET2}) becomes:
$$
\int_{X}L_{V}g^{\mu\nu}T_{\mu\nu} \mathrm{vol}_{g}. 
$$ 
A standard result from Riemannian geometry is that $L_{V}g^{\mu\nu} = \nabla^{\mu}V^{\nu} + \nabla^{\nu}V^{\mu}$, and since $T_{\mu\nu}$ is symmetric by definition, the above must be
$$
\int_{X} (\nabla^{\mu}V^{\nu})T_{\mu\nu} \mathrm{vol}_{g} = -\int_{X} V^{\nu}(\nabla^{\mu}T_{\mu\nu}) \mathrm{vol}_{g},
$$
where we invoked integration by parts  and the fact that $\nabla^{\mu}\mathrm{vol}_{g} = 0$ in the equality. Then, standard computations for generally covariant theories (which can be found in Appendix E of \cite{wald}) show that for on-shell fields (here meaning $\varphi$ such that $\Delta_{g}\varphi = 0$), the above integral is identically zero. For this to be true, it must be the case that 
\begin{equation}
	\nabla^{\mu}T_{\mu\nu} = 0.
\end{equation}
In the language of Noether's Theorem, the stress-energy tensor $T_{\mu\nu}$ is the conserved current associated to general covariance, a symmetry of a field theory coupled to a metric. 

In our regime, this implies that the conservation law $\nabla^{\mu}T_{\mu\nu} = 0$ is what is ultimately responsible for allowing us to define an $L_{\infty}$ action of the dg Lie algebra $\mathfrak{g}_{g}$ associated to the formal neighborhood of $[g] \in [\mathscr{M/D}]$ on observables $\mathrm{Obs}^{\mathrm{cl}}(X, \mathscr{L}_{g})$ for our generally covariant BV field theory in the sense of Definition \ref{Linfinityaction}. However, what we have shown above is only up to first order in the metric perturbation! The differential $\{S^{\mathfrak{g}},-\}$ we defined previously in principle contains ``higher conservation laws" associated to higher $L_{\infty}$ brackets read off from higher order terms in the power series $h \in T_{g}\mathscr{M}$. The author would be interested in assigning a physical interpretation to this. 

(2) Additionally, since the third and fourth terms are (up to a sign) the same as the first two, this means considering the first two alone should give us all the relevant data of the stress-energy tensor for a generally covariant field theory: we could even find a second order vector field equivariance property analogous to the one stated at the end of Remark \ref{equivariantfunctional} (we do just that in Section \ref{higherorders} of the Appendix).
	
\end{rmk}

\begin{const}
	Let us consider the ``infinitesimal general covariance" property more formally. Insight (1) suggests that the action functionals $S_{g}(\varphi)$ and 
	$$
	S_{g+\varepsilon L_{V}g}(\varphi) = \frac{-1}{2}\int_{X} \varphi \Delta_{g}\varphi \mathrm{vol}_{g} - \frac{\varepsilon}{2} \int_{X} L_{V}g^{\mu\nu} T_{\mu\nu} \mathrm{vol}_{g} =: S_{g}(\varphi) + \varepsilon I_{g}(L_{V}g, \varphi),
	$$
	where this equality holds modulo $\varepsilon^{2}$, should produce the same dynamics: this is true because for on-shell fields, the second term is zero. In other words, if we were to make sense of the differential $Q_{g + \varepsilon L_{V}g}$ for the BV space of fields, it should be appropriately equivalent to $Q_{g}$. Moreover, $Q_{g}$ induces the differential $\{S_{g}, -  \}$ on $\mathrm{Obs}^{\mathrm{cl}}(X, \mathscr{F}_{g})$, so that we would like $\{S_{g + \varepsilon L_{V}g}, -  \} = \{S_{g},- \} + \varepsilon\{I_{g}(L_{V}g), - \}$, the induced differential on $\mathrm{Obs}^{\mathrm{cl}}(X, \mathscr{F}_{g + \varepsilon L_{V}g})$ from $Q_{g + \varepsilon L_{V}g}$, to be similarly equivalent. To give all of the above hands and legs, we must rigorously define $\mathscr{F}_{g + \varepsilon L_{V}g}$ and its observables in the first place. 
	
	Let $\mathbb{D}_{2} = \mathbf{R}[\varepsilon]/(\varepsilon^{2})$ denote the (real) dual numbers. We can tensor $\mathscr{F}_{g} = C^{\infty}(X) \xrightarrow{Q_{g}} \mathrm{Dens}(X)[-1]$ with $\mathbb{D}_{2}$ to get $\mathscr{F}_{g} \otimes \mathbb{D}_{2}$, whose elements can be written as $\varphi_{0} + \varepsilon \varphi_{1}$ in degree 0 and similarly for degree 1. The differential $Q_{g + \varepsilon L_{V}g}$ looks like
	$$
	Q_{g} + \varepsilon D.
	$$
	It remains only to find $D$, which will depend on $g$ and $V$ and must be so that 
	\begin{center}
\begin{tikzcd}
\mathscr{F}_{g} \otimes \mathbb{D}_{2} = C^{\infty}(X) \otimes \mathbb{D}_{2} \arrow[r, "Q_{g} + \varepsilon 0"] \arrow[d, "\mathrm{Id} + \varepsilon L_{V}"]
& \mathrm{Dens}(X)[-1] \otimes \mathbb{D}_{2} \arrow[d, "\mathrm{Id} + \varepsilon L_{V}"] \\
\mathscr{F}_{g} \otimes \mathbb{D}_{2} = C^{\infty}(X) \otimes \mathbb{D}_{2} \arrow[r, "Q_{g} + \varepsilon D"]
& \mathrm{Dens}(X)[-1] \otimes \mathbb{D}_{2}
\end{tikzcd}
\end{center}
commutes. The downward-pointing arrows are $\mathrm{Id} + \varepsilon L_{V}$ since we are still assuming the diffeomorphism $f$ is generated by the vector field $V$: concretely, this is the first order approximation to the commuting square in Lemma \ref{scalequiv}. Thus, we are trying to suss out a neat form of the first-order perturbation of $Q_{g}$ with respect to the metric when the perturbation is along a diffeomorphism orbit. Our computations from Equation (\ref{infgencov}) suggest that we try $D = [L_{V}, Q_{g}]$. 
\end{const}

\begin{lem}\label{perturbfields}
	Let $\widetilde{\mathscr{F}}_{g} := (\mathscr{F}_{g} \otimes \mathbb{D}_{2}, Q_{g})$ and $\widetilde{\mathscr{F}}_{g + \varepsilon L_{V}g} := (\mathscr{F}_{g} \otimes \mathbb{D}_{2}, Q_{g} + \varepsilon [L_{V}, Q_{g}])$. Then the map $\mathrm{Id} + \varepsilon L_{V} : \widetilde{\mathscr{F}}_{g} \to \widetilde{\mathscr{F}}_{g + \varepsilon L_{V}g}$ is a cochain isomorphism (i.e. it is an equivalence of free BV field theories).
\end{lem}

\begin{rmk}
	We omit the proof: it is straightforward, albeit tedious. The above is the perturbative realization of general covariance: intuitively, the free BV scalar field coupled to a metric is equivalent to the free BV scalar coupled to an infinitesimally close metric in the same diffeomorphism orbit. This lemma also states that for the free scalar field with differential $Q_{g}$ on its BV space of fields, the first order deformation of $Q_{g}$ along the $\mathscr{D}$-orbit starting at $g \in \mathscr{M}$ is exactly 
	$$
	D = [L_{V},Q_{g}].
	$$ 
	This provides a nice coordinate-free form of the stress-energy tensor. 

\end{rmk}

\begin{rmk}
Such a lemma holds for any BV theory which is generally covariant by our definition: the only caveat is that the bookkeeping required to prove lemmas like those above may be more painstaking. The issues are that the Lie derivative $L_{V}$ manifests differently on different choices of fields, so one must be careful, and the bookkeeping may be more painstaking with higher $L_{\infty}$ terms. Additionally, certain BV fields have more than two terms in their cochain complexes: the computations in that case are more cumbersome, but only in the sense of needing to check multiple squares commute. This happens for example in Example \ref{perturbativeYM}.
\end{rmk}

Our goal was not only to show that these two ``infinitesimally close" spaces of fields were equivalent, but to show that their associated observables were similarly equivalent. This is what we do next. We need the following lemma:

\begin{lem}\label{sym}
If $\alpha : (V, d_{V}) \to (W, d_{W})$ is an isomorphism of cochain complexes, then there is an induced isomorphism $\alpha : (\mathrm{Sym}(V), d_{V}) \to (\mathrm{Sym}(W), d_{W})$ of cochain complexes, where the differentials $d_{V}$ and $d_{W}$ are extended to the respective symmetric algebras as derivations.
\end{lem}

\begin{rmk}
It is similarly true that $\mathrm{Sym}(V^{\vee})$ and $\mathrm{Sym}(W^{\vee})$ are isomorphic cochain complexes: the differentials on $V^{\vee}$ and $W^{\vee}$ are induced by those on $V$ and $W$, and using this lemma once more gives $(\mathrm{Sym}(V^{\vee}), d_{V}) \cong (\mathrm{Sym}(W^{\vee}), d_{W})$. (We abuse notation so that $d_{V}$ and $d_{W}$ are the differentials induced from those on $V$ and $W$, respectively.)
\end{rmk}

One might expect that because the na\"{i}ve algebraic symmetric powers of $\mathscr{F}_{g}$ are not what we use to define observables, we should be wary; however, the completed projective tensor product we used to define functionals is the necessary one in the case of infinite-dimensional vector spaces for these constructions to carry over. We can now state a key theorem:

\begin{thm}\label{akeythm}
	We have the following isomorphism of classical observables:
	\begin{equation}
	\mathrm{Obs}^{\mathrm{cl}}(X,\widetilde{\mathscr{F}}_{g}) \cong \mathrm{Obs}^{\mathrm{cl}}(X, \widetilde{\mathscr{F}}_{g + \varepsilon L_{V}g}),
	\end{equation}
	where the isomorphism is induced by the isomorphism $\mathrm{Id} + \varepsilon L_{V}$ from Lemma \ref{perturbfields}. 
\end{thm}

	\begin{proof}
		Since Lemma \ref{sym} holds for infinite dimensional cochain complexes with the definition of $\mathrm{Sym}(\widetilde{\mathscr{F}}_{g})$ as in Definition \ref{functionals} (i.e. with the completed projective tensor product), we indeed have that the isomorphism $\mathrm{Id} + \varepsilon L_{V}$ from Lemma \ref{perturbfields} induces an isomorphism of $(\mathscr{O}(\widetilde{\mathscr{F}}_{g}), \{S_{g}, - \})$ and $(\mathscr{O}(\widetilde{\mathscr{F}}_{g + \varepsilon L_{V}g}), \{S_{g}, - \} + \varepsilon \{I_{g}(L_{V}g), - \})$.  This is the result.
	\end{proof}

		Recall that although we have done the precise computations in the case of the massless free scalar field, the same statement holds in the case of any BV theory with differential $Q_{g}$.

\begin{rmk}
As mentioned earlier, $\mathfrak{g}_{g}$ is a sheaf on $\mathbf{Riem}_{n}$, and so identical computations as in Theorem \ref{akeythm} imply the analogous equivalence of factorization algebras for the equivariant observables of Proposition \ref{themainthm}.
\end{rmk}

\begin{rmk}
	This result follows almost directly from a theory exhibiting general covariance; however, having isomorphisms written down explicitly and recognizing their naturality when compared to the non-perturbative definition of general covariance provides a sanity check, not to mention an enhanced perspective on quantities like the stress-energy tensor.
	
Checking this theorem over a fixed $g \in \mathscr{M}$ and invoking $\mathscr{D}$-equivariance implies that $\{S_{g}, - \}$ is the differential over the entire diffeomorphism orbit $\mathscr{D} \cdot g \subset \mathscr{M}$. Similarly, seeing how $\{S_{g},-\}$ varies over a formal neighborhood of $g$ (i.e. expanding $\{S_{g+ \varepsilon h},-\}$ in consecutive orders of $\varepsilon h$) really grants us a view of the formal neighborhood of all of $\mathscr{D} \cdot g$: this is precisely equivalent to considering the formal neighborhood of $g$ as an element of the quotient stack $[\mathscr{M/D}]$. 
\end{rmk}

\begin{rmk}
Note that although Theorem \ref{akeythm} is computed on a fixed $X$ for simplicity, it holds at the level of factorization algebras, in the sense that $\mathrm{Obs}^{\mathrm{cl}}(-,\widetilde{\mathscr{F}}_{g})$ and  $\mathrm{Obs}^{\mathrm{cl}}(-, \widetilde{\mathscr{F}}_{g + \varepsilon L_{V}g})$ define equivalent factorization algebras $\mathbf{Riem}_{n} \to \mathbf{dgVect}$.
\end{rmk}

A few remarks on higher order versions of this theorem are made in Appendix \ref{higherorders}.

\section{Appendix}

\subsection{A detailed example}\label{findim}
The following is a detailed example of how Chevalley-Eilenberg cochains arise as functions on a formal neighborhood around a point in a stack: it is meant to supplement what was shown in Lemma \ref{CECs}. Fix coordinates $(x_{1}, \ldots, x_{n})$ on $\mathbf{R}^{n}$ and consider an action $P: G \to \textrm{Diff}(\mathbf{R}^{n})$ for a finite-dimensional Lie group $G$. The total derivative of this map is $\rho : \mathfrak{g} \to \textrm{Vect}(\mathbf{R}^{n}) \cong C^{\infty}(\mathbf{R}^{n}) \otimes \mathbf{R}^{n}$, which for $\alpha \in \mathfrak{g}$ has some coordinate expression: 
$$
\alpha \mapsto \sum_{i=1}^{n} f(x_{i}, \alpha)\partial_{i},
$$
where we use the shorthand $\partial / \partial x_{i} = \partial_{i}$. If we restrict to a formal neighborhood of the origin, $\widehat{\mathbf{R}}^{n}_{0}$, and compute its space of functions, we get the usual Taylor series of functions about the origin, $C^{\infty}(\widehat{\mathbf{R}}^{n}_{0}) \cong \widehat{\textrm{Sym}}(T_{0}^{\vee}\mathbf{R}^{n}) \cong \mathbf{R}\llbracket x_{1}, \ldots, x_{n} \rrbracket$, which we will denote $\mathbf{R} \llbracket \mathbf{x} \rrbracket$ when convenient. Thus, restricting the preceding derivative to the formal neighborhood of $0$ gives us $\rho_{0} : \mathfrak{g} \to \textrm{Vect}(\widehat{\mathbf{R}}^{n}_{0}) \cong \mathbf{R} \llbracket \mathbf{x} \rrbracket \otimes \widehat{\mathbf{R}}^{n}_{0}$, which looks like: 
$$
\alpha \mapsto \sum_{i=1}^{n} \hat{f}_{0}(x_{i}, \alpha)\partial_{i},
$$
where $\hat{f}_{0}$ denotes the Taylor expansion of $f$ at $0$. This defines an action of $\mathfrak{g}$ on $\mathbf{R} \llbracket \mathbf{x} \rrbracket$ by derivations, and so we can thus define $C^{\bullet}(\mathfrak{g}, \mathbf{R} \llbracket \mathbf{x} \rrbracket)$.

Fixing a basis $\{ \alpha_{1}, \ldots, \alpha_{m} \}$ for $\mathfrak{g}$ (assuming finite dimension $m$), denote the dual basis for $\mathfrak{g}^{\vee}$ as $\{ \alpha^{1}, \ldots, \alpha^{m} \}$. With these coordinates, we can write $C^{\bullet}(\mathfrak{g}, \mathbf{R} \llbracket \mathbf{x} \rrbracket)$ as $\mathbf{R} \llbracket \alpha^{1}, \ldots, \alpha^{m}, x_{1}, \ldots, x_{n} \rrbracket$, where the $\alpha^{k}$ are in degree $1$ and the $x_{k}$ in degree $0$. Thus, it is sufficient to see what the differential $d_{CE}$ does on an element of the form $\alpha^{k} \otimes x_{l}$, for $\alpha^{k} \in \mathfrak{g}^{\vee}[-1]$ and $x_{l} \in \mathbf{R} \llbracket \mathbf{x} \rrbracket$, to classify its behavior. Momentarily viewing $\alpha^{k}$ as a degree $1$ element of just $C^{\bullet}(\mathfrak{g}) = \textrm{Sym}(\mathfrak{g}^{\vee}[-1])$, and noting that $d_{CE} : \mathfrak{g}^{\vee}[-1] \to \textrm{Sym}^{2}(\mathfrak{g}^{\vee}[-1])$ is dual to the bracket $[-,-] : \textrm{Sym}^{2}(\mathfrak{g}^{\vee}[-1]) \to \mathfrak{g}^{\vee}[-1]$, we have:
$$
d_{CE}\alpha^{k} = \frac{-1}{2} \sum_{i,j = 1}^{m} c^{k}_{ij} \alpha^{i} \wedge \alpha^{j},
$$
where $c^{k}_{ij}$ are the structure constants for $\mathfrak{g}$. Concurrently, for $x_{l} \in \mathbf{R} \llbracket \mathbf{x} \rrbracket$,
$$
d_{CE}x_{l} = \sum_{i=1}^{m} \alpha^{i} \otimes \alpha_{i} \cdot x_{l}.
$$
Therefore, by requiring the usual derivation rules, we get,
$$
d_{CE}(\alpha^{k} \otimes x_{l}) = \frac{-1}{2} \sum_{i,j = 1}^{m} c^{k}_{ij} \alpha^{i} \wedge \alpha^{j} \otimes x_{l} + \sum_{i = 1}^{m} \alpha^{k} \wedge \alpha^{i} \otimes \alpha_{i} \cdot x_{l},
$$
which we extend to the rest of $C^{\bullet}(\mathfrak{g}, \mathbf{R} \llbracket \mathbf{x} \rrbracket)$ with the Leibniz rule. A coordinateless way of writing this is $d_{CE} = [-,-]_{\mathfrak{g}}^{\vee} + \rho^{\vee}_{0}$, where $\rho^{\vee}_{0}$ encodes a dual to the action map $\rho_{0} : \textrm{Vect}(\widehat{\mathbf{R}}^{n}_{0}) \to \mathfrak{g}^{\vee} \otimes \textrm{Vect}(\widehat{\mathbf{R}}^{n}_{0})$ as described implictly above. In this example, $d$ is in fact a vector field, specificially
\begin{equation}
    d_{CE} = \frac{-1}{2}c^{k}_{ij}\alpha^{i} \wedge \alpha^{j} \frac{\partial}{\partial \alpha^{k}} + \alpha^{i} \otimes (\alpha_{i} \cdot x_{l}) \frac{\partial}{\partial x_{l}},
\end{equation}
on the formal neighborhood of $0$ in the stack $[\mathbf{R}^{n}/G]$, where we have used the Einstein summation convention over repeated indices in the last step.

To make things easier to grasp, let us consider the case of $SO(2)$ acting on $\mathbf{R}^{2}$ via rotations. Then $\mathfrak{g} = \mathfrak{so}(2)$ and the Taylor series ring about the origin is $\mathbf{R} \llbracket x,y \rrbracket$. The representation map $\rho_{0} : \mathfrak{so}(2) \to \textrm{Der}(\mathbf{R} \llbracket x,y \rrbracket) \cong \mathbf{R} \llbracket x,y \rrbracket \otimes \mathbf{R}^{2}$ is 
$$
\begin{pmatrix}
0 & -1\\
1 & 0
\end{pmatrix}  \mapsto  y\partial_{x} - x\partial_{y},
$$ 
from which we can define $C^{\bullet}(\mathfrak{so}(2), \mathbf{R} \llbracket x,y \rrbracket)$. We will leave it as an exercise to the reader to show that $H^{0}(\mathfrak{so}(2), \mathbf{R} \llbracket x,y \rrbracket)$ is the set of rotation-invariant Taylor series around $0$. This is not surprising: more generally, in the case of $SO(n)$ acting on $\mathbf{R}^{n}$, $H^{0}(\mathfrak{so}(n), \mathbf{R} \llbracket x_{1}, \ldots, x_{n} \rrbracket)$ is the set of $SO(n)$-invariant Taylor series around the origin in $\mathbf{R}^{n}$. 

Another enlightening exercise is to consider the appropriate Chevalley-Eilenberg cochains coming from formal neighborhoods of points away from the origin; e.g. the zeroth cohomology group of the cochains around $(x_{0}, 0)$ is isomorphic to $\mathbf{R} \llbracket x-x_{0} \rrbracket$. Notice there that the vector fields coming from the action at these non-fixed points have constant coefficient terms.

\subsection{A remark on higher orders}\label{higherorders}
We would like to make sense of Theorem \ref{akeythm} in the case that we do not cut off the orders of $\varepsilon$ after only a linear perturbation. The linear perturbation gives the necessary data to understand the stress-energy tensor in the usual way; however retaining higher orders of $\varepsilon$ to compute ``higher" stress-energy tensors may be relevant, and the BV formalism gives an ideal way of interpreting and packaging that data. To put it plainly, we'd like to expand $\{S_{g+ \varepsilon h},-\}$ in more powers of $\varepsilon h$. 

\begin{const}
A concrete jumping-off point here would be to consider that for a generally covariant theory, we have 
	\begin{equation}\label{infgencov2}
\frac{d^{k}}{dt^{k}}\int_{X} (f_{t}^{*}\varphi)\Delta_{	f_{t}^{*}g} (f_{t}^{*}\varphi)\mathrm{vol}_{f_{t}^{*}g} \Big|_{t=0} = 0
\end{equation}
for any $k > 0$. This is the general form of Equation (\ref{infgencov}). For now, let us stick with $k=2$. The above should have an analogous unpacking to the one following Equation (\ref{infgencov}); however, we then need to make sense of 
$
\frac{d^{2}}{dt^{2}}(f_{t}^{*}\varphi) \big|_{t=0}. 
$
A short exercise in differential geometry gives us that  
\begin{equation}\label{LVLV}
\frac{d^{2}}{dt^{2}}(f_{t}^{*}\varphi) \big|_{t=0} = \frac{1}{2} L_{V} (L_{V}\varphi),
\end{equation}
as one might expect; the same equation holds for any $k > 0$, and the right side is in fact equal to $\frac{1}{k!}L_{V}^{k}\varphi$, where $L_{V}^{k}$ denotes taking the Lie derivative with respect to the vector field $V$ $k$ times. 

We can now consider a similar computation to the one in Lemma \ref{perturbfields}, replacing $\mathbb{D}_{2}$ with $\mathbb{D}_{3} := \mathbf{R}[\varepsilon]/(\varepsilon^{3})$ and using the above identity, to figure out what operator $D_{2}$ makes the following square commute: 

\[\begin{tikzcd}
	{\mathscr{F}_{g} \otimes \mathbb{D}_{3} = C^{\infty}(X) \otimes \mathbb{D}_{3}} &&&&& {\mathrm{Dens}(X)[-1] \otimes \mathbb{D}_{3}} \\
	\\
	\\
	{\mathscr{F}_{g} \otimes \mathbb{D}_{3} = C^{\infty}(X) \otimes \mathbb{D}_{3}} &&&&& {\mathrm{Dens}(X)[-1] \otimes \mathbb{D}_{3}}
	\arrow["{\mathrm{Id} + \varepsilon L_{V} + \frac{\varepsilon^{2}}{2} L_{V}^{2}}", from=1-1, to=4-1]
	\arrow["{Q_{g} + \varepsilon 0 + \varepsilon^{2}0}", from=1-1, to=1-6]
	\arrow["{\mathrm{Id} + \varepsilon L_{V} + \frac{\varepsilon^{2}}{2} L_{V}^{2}}", from=1-6, to=4-6]
	\arrow["{Q_{g} + \varepsilon D_{1} + \varepsilon^{2}D_{2}}", from=4-1, to=4-6],
\end{tikzcd}\]
where we have renamed $D = D_{1}$ from above to emphasize the order of $\varepsilon$ it is associated to. The result is the following:
\end{const}	

\begin{lem}
The above square commutes if we choose
$$
D_{2} = \frac{1}{2}[L_{V}^{2}, Q_{g}] - [L_{V}, Q_{g}]L_{V}.
$$	
Moreover, $\widetilde{\mathscr{F}}_{g} := (\mathscr{F}_{g} \otimes \mathbb{D}_{3}, Q_{g})$ and $ \widetilde{\mathscr{F}}_{g + \varepsilon L_{V}g} := (\mathscr{F}_{g} \otimes \mathbb{D}_{3}, Q_{g} + \varepsilon [L_{V}, Q_{g}] + \varepsilon^{2}(\frac{1}{2}[L_{V}^{2}, Q_{g}] - [L_{V}, Q_{g}]L_{V}) )$, are cochain isomorphic via the map $\mathrm{Id} + \varepsilon L_{V} + \frac{\varepsilon^{2}}{2}L_{V}^{2}$. 
\end{lem}

\begin{rmk}
	The operator $D_{2} = \frac{1}{2}[L_{V}^{2}, Q_{g}] - [L_{V}, Q_{g}]L_{V}$ thus represents a sort of ``higher" stress-energy tensor for a generally covariant theory, in the same way that $D_{1} = [L_{V},Q_{g}]$ did so in the first order case. It also satisfies some conservation property (analogous to $\nabla^{\mu}T_{\mu\nu} = 0$): otherwise, we would not have this cochain isomorphism of field theories. However, it would be harder to pin down a physical interpretation of the associated conservation law. 
	\end{rmk}
	
	\begin{rmk}
	Additionally, we could now update Theorem \ref{akeythm} so that it holds up to second order in $\varepsilon$: the isomorphism of observables is induced from the isomorphism $\mathrm{Id} + \varepsilon L_{V} + \frac{\varepsilon^{2}}{2}L_{V}^{2}$ of the field theories. The proof is otherwise the same. It may be clear to the reader by now that these results can be generalized to arbitrarily high orders of $\varepsilon$. In that case, we can expand the differential as
	$$
	Q_{g + \varepsilon L_{V}g} = Q_{g} + \varepsilon [L_{V}, Q_{g}] + \varepsilon^{2} (\frac{1}{2}[L_{V}^{2}, Q_{g}] - [L_{V}, Q_{g}]L_{V}) + \varepsilon^{3}D_{3} + \ldots .  
	$$
	on the fields--as long as the metric perturbation is induced by a vector field--and pick out $D_{k}$ for all $k \in \mathbf{N}$ so that we get an analogous commutative square, with the isomorphism 
	$$
	\mathrm{Id} + \sum_{k=1}^{\infty} \frac{\varepsilon^{k}}{k!}L_{V}^{k}.
	$$ 
Thus, the isomorphism of observables 
$
	\mathrm{Obs}^{\mathrm{cl}}(X,\widetilde{\mathscr{F}}_{g}) \cong \mathrm{Obs}^{\mathrm{cl}}(X,\widetilde{\mathscr{F}}_{g + \varepsilon L_{V}g})
$
remains true to all orders of $\varepsilon$, since we see from the above that all of the appropriate $D_{k}$ must exist, regardless of how difficult they are to compute or interpret physically. To see more explicitly, one would need to consider expansions of $\{S_{g + \varepsilon h}, - \}$ to all orders of $\varepsilon h$: a first step in this case would be to generalize Lemma \ref{derlaplace} to higher derivatives with respect to $t$.

\begin{rmk}
As well as checking for agreement with a generally covariant theory when $h = L_{V}g$ as we did above, we may want to consider perturbations in the direction of various geometric flows. For example, it would be fruitful to consider field theories over metrics related by the Ricci flow, and a first step here would be to perturb a fixed metric $g$ in the direction of that flow. Ricci flow $\partial_{t}g = -\mathrm{Ric}(g)$ is itself a ``generally covariant flow" in the sense that the equation is diffeomorphism equivariant:\footnote{This flow is in fact the gradient flow of the diffeomorphism equivariant Einstein-Hilbert action functional.} I would be interested to see how viewing it as a flow on the moduli stack of metrics modulo diffeomorphism would provide some advantages.
\end{rmk}
	
\end{rmk}

\subsection{Acknowledgements} My gratitude primarily goes to my advisor Owen Gwilliam for suggesting this course of study and for steadfastly supporting me during the writing and revision process. I would also like to thank Eugene Rabinovich for his patience in working out many details in the BV framework with me, Andreas Hayash for being a great interlocutor in the subject of stacks, and Nicholas Teh for actively engaging with me about the content of this paper. Finally, many thanks are due to the anonymous referees for helping me greatly improve this work, and for motivating a broadening of my knowledge generally.

\subsection{Statements and Declarations} \textit{Conflict of Interest Statement:} On behalf of all authors, the corresponding author states that there is no conflict of interest. \\
\textit{Data Availability Statement:} My manuscript has no associated data.

\end{document}